\documentclass[fleqn,10pt]{wlscirep}
\usepackage[utf8]{inputenc}
\usepackage[T1]{fontenc}
\usepackage{amsfonts,amscd,amsmath,amsthm}
\usepackage{epsfig}
\usepackage{subfigure}
\usepackage{xcolor}
\usepackage[colorlinks = true]{hyperref}
\usepackage{physics}
\usepackage{epstopdf}
\usepackage{framed}
\usepackage{multirow}
\usepackage{color}
\usepackage{longtable}
\usepackage{comment}
\usepackage[ruled,vlined]{algorithm2e}
\usepackage[most]{tcolorbox}

\usepackage{tikz}
\usetikzlibrary{tikzmark, calc, fit, positioning}
\usetikzlibrary{shapes}
\usetikzlibrary{quantikz}

\newtheorem{theorem}{Theorem}
\newtheorem{lemma}{Lemma}
\newtheorem{corollary}{Corollary}

\newtheorem{definition}{Definition}
\newtheorem{proposition}{Proposition}

\graphicspath{ {./figs/} }

\usepackage{pifont}
\usepackage{bm}

\newcommand{\cL}{\mathcal{L}}
\newcommand{\cA}{\mathcal{A}}

\newcommand{\cB}{\mathcal{B}}
\newcommand{\PTM}{\text{PTM}}

\newcommand{\cV}{\mathcal{V}}
\newcommand{\cU}{\mathcal{U}}
\newcommand{\bx}{\mathbf{x}}

\newtheorem*{theorem*}{Theorem}

\newtheorem{manualtheoreminner}{Theorem}
\newenvironment{manualtheorem}[1]{%
	\renewcommand\themanualtheoreminner{#1}%
	\manualtheoreminner
}{\endmanualtheoreminner}
\newcommand{\revise}[1]{{#1}}

\title{Limitations of Noisy Quantum Devices in Computing and Entangling Power}

\author[1]{Yuxuan Yan}
\author[1]{Zhenyu Du}
\author[1]{Junjie Chen}
\author[1,*]{Xiongfeng Ma}
\affil[1]{Center for Quantum Information, Institute for Interdisciplinary Information Sciences, Tsinghua University, Beijing 100084, P.~R.~China}
\affil[*]{xma@tsinghua.edu.cn}

\begin{abstract}
    Finding solid and practical quantum advantages via noisy quantum devices without error correction is a critical but challenging problem. Conversely, comprehending the fundamental limitations of the state-of-the-art is equally crucial. In this work, we consider the class of strictly contractive unital noise and derive its analytical representation by decomposition. Under such noise, we observe the polynomial-time indistinguishability of $n$-qubit devices from random coins when circuit depths exceed $\Omega(\log(n))$. Even with classical processing, we demonstrate the absence of computational advantage in polynomial-time algorithms with super-logarithmic noisy circuit depths. These results impact variational quantum algorithms, error mitigation, and quantum simulation with polynomial depth. Furthermore, we consider noisy quantum devices with a restricted gate topology. For one-dimensional noisy qubit circuits, we rule out super-polynomial quantum advantages in all-depth regimes. We also establish upper limits on entanglement generation: $O(\log(n))$ for one-dimensional circuits and $O(\sqrt{n} \log(n))$ for two-dimensional circuits. Our findings underscore the computational capacity and entanglement scalability constraints in noisy quantum devices.
\end{abstract}

\begin{document}

\flushbottom
\maketitle

\thispagestyle{empty}

\section*{Introduction}

Recent advancements in quantum computing have notably enhanced the scale and fidelity of quantum devices, surpassing classical brute-force simulation capabilities \cite{arute_quantum_2019, wu_strong_2021}. Meanwhile, scalable quantum error correction remains out of reach, primarily due to high noise levels and the limited qubit count. Consequently, current quantum devices, ranging from dozens to several hundred noisy qubits, are situated between classical computing and fault-tolerant quantum computing, a regime known as the Noisy Intermediate-Scale Quantum (NISQ) era \cite{preskill_quantum_2018}. Despite ongoing progress towards quantum error correction \cite{egan_fault-tolerant_2021,gong_experimental_2022,postler_demonstration_2022,abobeih_fault-tolerant_2022,zhao_realization_2022,google_quantum_ai_suppressing_2023, ni_beating_2023, bluvstein_logical_2023}, transitioning from NISQ to fault-tolerant quantum computation presents a formidable challenge, which is expected to take years or even decades.

	Within the NISQ paradigm, many important previous works have been devoted to the search for advantages over classical computers without error correction. Here, ``advantages'' refer to the quantum device's capability to accelerate computational tasks beyond purely classical means. In this line of research, specific problems demonstrate theoretical quantum advantages in shallow noisy circuit regimes \cite{bravyi_quantum_2020} or in oracle settings \cite{chen_complexity_2023}. In reality, experiments of noisy quantum circuit sampling challenge the ability of the most powerful classical computers \cite{arute_quantum_2019, wu_strong_2021}. However, in more practical applications, the noise in quantum devices often undermines advantages. Notably, as important candidates for practical advantage, variational quantum algorithms exhibit severe fragility to noise \cite{stilck_franca_limitations_2021, de_palma_limitations_2023}. These findings trigger active research on more noise-resilient algorithms \cite{temme_error_2017, endo_practical_2018, cai_quantum_2023} and, in the meantime, lead to a fundamental question: What is the limit of the noisy device's power in the NISQ era, applicable to generic computations?

	To properly address this question, we need an in-depth consideration of the classical computation part in quantum algorithms. In many cases, the cooperation between quantum and classical computations is significant for our understanding of the power of NISQ devices. When viewed only in itself, a noisy quantum device experiences a rapid loss of information due to noise and loses its computational power \cite{aharonov_limitations_1996, kastoryano_quantum_2013}. In reality, we often use classical computers to assist quantum devices and perform computations that might be difficult for quantum computers, thereby easing the limitations of noisy quantum devices. In the NISQ era, classical enhancements have been emphasized for their functionality in mitigating noise and amplifying the capabilities of quantum hardware \cite{endo_practical_2018, endo_hybrid_2021}. The problem of assessing NISQ advantages, with classical computer enhancements considered, is complicated and requires a systematic understanding.

	The classical simulatability is also a crucial factor in studying NISQ advantages \cite{markov_simulating_2008, noh_efficient_2020, cheng_simulating_2021, napp_efficient_2022, cheng_efficient_2023, aharonov_polynomial-time_2023}, as easier simulation implies weaker advantages. Some of the classical simulation algorithms have theoretical guarantees under certain conditions, such as when having small underlying graph treewidth \cite{markov_simulating_2008} or being in the anti-concentration regime under noise \cite{aharonov_polynomial-time_2023}. Under these conditions, polynomial-time classical algorithms exist, excluding super-polynomial quantum advantages. More generally, the classical simulatability and advantages of noisy quantum devices remain an open question.

	In this work, we establish a clear boundary for the limitations of noisy quantum devices. Our analysis is based on the strictly contractive unital noise model without error correction. For such a noise model, we present an analytical representation to characterize its behavior. Under such noise, we show the statistical indistinguishability of quantum device outputs from a uniform distribution when the circuit depth exceeds the logarithm of the running time. When incorporating any form of classical processing, we show that devices with super-logarithmic circuit depths $\Omega(\log(n))$ fail to deliver any quantum advantage for polynomial-time quantum algorithms. We also unify and strengthen previous findings in the limitations of NISQ algorithms, including variational quantum algorithms \cite{mcardle_quantum_2020, cerezo_variational_2021} and error mitigation \cite{temme_error_2017, endo_practical_2018, cai_quantum_2023}.

	Further, we extend our results to platforms under connectivity constraints, such as superconducting qubits \cite{arute_quantum_2019, wu_strong_2021}. In these cases, spatial locality further challenges the quantum advantages of noisy devices. At any depth, we show that one-dimensional (1D) noisy quantum circuits are simulable, which suggests the absence of super-polynomial advantages. To characterize limitations from a more physical perspective, we establish limits of the maximal entanglement generation in noisy quantum devices---for a 1D qubit array, the capacity to generate quantum entanglement is capped at $O(\log(n))$; for two-dimensional (2D) case, $O(\sqrt{n} \log(n))$. The results rule out the efficient creation of highly entangled states, such as highly excited or thermalized states of almost all quantum systems. Our findings are summarized in Fig.~\ref{fig:summary}.

	\begin{figure}[bthp!]
    \centering
    \subfigure[]{
    \raisebox{0cm}{ 
        \begin{minipage}[t]{\linewidth} 
            \centering
            \begin{tikzpicture}[
                scale=1,
                >=latex
                ]
                \def\height{8}   
                \def\trange{8} 

                \coordinate (O) at (0,0);
                \draw[->,thick] (O) --++ (\height,0) node[below] {Size $n$};
                \draw[->,thick] (O) --++ (0,\trange) node[above] {Depth};
                \fill[green!30] (\height-0.2, 0) -- plot[domain=0:\trange, variable=\t] ({\fpeval{ 2.5*exp(\t*0.14) }}, {\t}) -- (\height-0.2, \trange) -- cycle;
                \draw[samples=400,scale=1,domain=0:\trange,smooth,variable=\t,black] plot ({\fpeval{ 2.5*exp(\t*0.14) }}, {\t});
				\draw[samples=400,scale=1,domain=0:\trange,dashed,variable=\t,black] plot ({\fpeval{ 1.2+4.5*\t*0.14}}, {\t}); 
                \draw[samples=400,scale=1,domain=0:\trange,dashed,variable=\t,black] plot ({\fpeval{ 0.2+3.9*sqrt(\t*0.14)}}, {\t}); 
                \draw[samples=400,scale=1,domain=0:\trange,dashed,variable=\t,black] plot ({\fpeval{ 0.1+3.4*ln(1+\t*0.1)}}, {\t}); 
				\node[anchor=south] at ({\fpeval{ 2.5*exp(\height*0.14) }}, {\height}) {$\Omega(\frac{\log(n)}{|\log(\mu)|})$};
				\node[anchor=south] at ({\fpeval{ 1.2+4.4*\height*0.14}}, {\height}) {HHL $O(n)$};
                \node[anchor=south] at ({\fpeval{ 0.2+3.9*sqrt(\height*0.14)}}, {\height}) {Shor $O(n^2)$};
                \node[anchor=south] at ({\fpeval{ 0.1+3.4*ln(1+\height*0.1)}}, {\height}) {Grover $\exp(O(n))$};
                \node at (6, 3.5) {Possible advantage};
                \node at (6, 1.5) {\scalebox{8}{\textit{?}}};
                \node at (2.5, 5.5) { \color{orange} \scalebox{8}{\ding{55}} };
            \end{tikzpicture}
        \end{minipage}
    }
	}
    \subfigure[]{
        \begin{minipage}[t]{\linewidth}
            \centering
            \begin{tikzpicture}[
                scale=1,
                >=latex
                ]
                \def\height{5}
                \def\width{11.2}
                \coordinate (O) at (0,0);
                \draw[->,thick] (O) --++ (\width,0) node[below] {Topology};
                \draw[->,thick] (O) --++ (0,\height) node[above] {Depth};
                \node[thick,rectangle,draw,minimum width=4.5cm, minimum height=2.5cm,align=center,anchor=south west] (Zone1) at (O) {
                    Classial Simulatable \\[0.15cm]
                    \scalebox{1}{ Complexity $\sim n^{O(1/|\log(\mu)|)}$} \\
                    \scalebox{1}{ Entanglement $\sim O(\log(n))$}};

                \node[thick,rectangle,draw,fill=green!10,minimum width=3.5cm, minimum height=2.5cm,align=center,anchor=south west] (Zone2) at (Zone1.south east) {
                    Entanglement \\[0.15cm] $\sim O(\sqrt{n}\log(n))$};

                \node[thick,rectangle,draw,fill=green!50,minimum width=3cm, minimum height=2.5cm,align=center,anchor=south west] (Zone3) at (Zone2.south east) {\scalebox{4}{\textit{?}}};

                \node[thick,rectangle,draw,minimum width=11.05cm, minimum height=2cm,align=center,anchor=south west] (Zone4) at (Zone1.north west) {
                    Too noisy \\[0.15cm]
                    (like random coins) \\
                    {\color{orange} \scalebox{4}{\ding{55}}}
                };
                \node[left] at (0, 2.5) {$\Omega(\frac{\log(n)}{|\log(\mu)|})$};
                \node[below=2pt] at (Zone1.south) {1D};
                \node[below=2pt] at (Zone2.south) {2D};
            \end{tikzpicture}
        \end{minipage}%
    }
    	\caption{\revise{\textbf{Summary of the limitations for noisy quantum devices without error correction.} \textbf{a} For algorithms with generic classical processing control, we prove that devices with circuit depths beyond $\Omega(\log(n))$ under single-qubit depolarizing noise are too noisy to offer any computational advantage in a polynomial running time, including well-known quantum algorithms, such as Shor's, Grover's, and the Harrow-Hassidim-Lloyd algorithm. The scaling of depth with the qubit number, $n$, is shown by dashed curves for each algorithm, with a solid curve showing the logarithmic upper limit. The scaling is shown in an asymptotic limit, i.e., when $n$ is large. \textbf{b} In the regime where the circuit depth is below the logarithmic scaling, potential quantum advantages depend on the gate connection topology of noisy quantum devices. For the 1D case, we prove classical simulatability for noisy devices and give an entanglement upper bound of $O\left(\log (n)\right)$. For the 2D case, the entanglement generation upper bound scales as $O(\sqrt{n} \log (n))$. Therefore, super-polynomial advantages without error correction are only possible when gate connectivity is higher than one dimension, and the circuit depth is below logarithmic scaling for noisy devices. Such a regime is colored green.}}
		\label{fig:summary}
\end{figure}

	\section*{Results}

	In the NISQ era, we represent noisy device operations as a sequence of unitary gates interspersed with independent single-qubit noise $\Lambda_1$, assumed as a strictly contractive unital channel.
	\revise{
		Here, the contractive property for unital channels is characterized by the contractive rate,
		\begin{equation}
			\mu_1 = \sup_{\substack{\rho \in \mathcal{D}(\mathcal{H}) \\ 0 < D(\rho \| \frac{I}{2}) < \infty}}
			\frac{D(\Lambda_1(\rho) \| \frac{I}{2})}{D(\rho \| \frac{I}{2})},
		\end{equation}
		where $\mathcal{H}$ denotes the single-qubit Hilbert space, $\mathcal{D}(\mathcal{H})$ denotes its dual space, and $D(\rho \| \sigma) = \text{Tr}(\rho \log \rho - \rho \log \sigma)$ denotes the quantum relative entropy.
	}
	For strictly contractive unital channels, we will have $\mu_1 < 1$.
	Such noise model will strictly drive any single-qubit state $\rho$ towards a maximally mixed state, $D(\Lambda_1(\rho) \| \frac{I}{2}) \le \mu_1 D(\rho \| \frac{I}{2})$.

	This class of noise encompasses a broad range of unital channels. While their contractive property may initially seem abstract, we provide a direct analytical representation to clarify their behavior, as in the following proposition. The derivation of this proposition is available in Methods.

	\begin{proposition} \label{prop:noise-decompose}
		A single-qubit unital channel $\Lambda_1$ is strictly contractive \emph{if and only if} it can be decomposed as
	\begin{equation} \label{eq:analytic_expression}
		\Lambda_1 = \mathcal{V} \circ \mathcal{D} \circ \mathcal{U},
	\end{equation}
	where $\mathcal{U}$ and $\mathcal{V}$ are unitary channels, and $\mathcal{D}$ is a channel that contracts each Pauli matrix $\sigma_i$ by a factor $q_i < 1$, which corresponds to $\mathrm{Diag}(1, q_X, q_Y, q_Z)$ in the Pauli-Liouville representation.
	\end{proposition}

	A canonical example is the depolarizing channel, defined as $\Lambda_1(\rho) = (1-p)\rho + p \frac{I}{2}$, for which $q_X = q_Y = q_Z = 1 - p$, and the contraction rate is $\mu_1 = (1 - p)^2$~\cite{muller-hermes_relative_2016}.

	After gate operations with associated noise, computational-basis measurements yield the output. The noisy device model is formally defined as follows, with a visual representation in Fig.~\ref{fig:ndepol}.

    \begin{figure}[bthp!]
		\centering
		\begin{minipage}[t]{0.6\linewidth}
			\subfigure[]{
				\centering
				\begin{quantikz}[row sep=0.2cm,column sep=0.2cm]
					\lstick{$\ket{0}$} & \gate[4, nwires=3]{\mathcal{U}_1} & \phase{} & \gate[4, nwires=3]{\mathcal{U}_2} & \phase{} & \ \ldots\ \qw      & \gate[4, nwires=3]{\mathcal{U}_t} & \phase{} & \meter{} \\
					\lstick{$\ket{0}$} &                                   & \phase{} &                                   & \phase{} & \ \ldots\ \qw      &                                   & \phase{} & \meter{} \\
					\lstick{\vdots\ }  &                                   &          &                                   &          &  \lstick{\vdots\ } &                                   &          & \lstick{\vdots\ } \\
					\lstick{$\ket{0}$} &                                   & \phase{} &                                   & \phase{} & \ \ldots\ \qw      &                                   & \phase{} & \meter{} \\
				\end{quantikz}
				\label{fig:ndepol}
			}
			\subfigure[]{
				\includegraphics[width=0.85\linewidth]{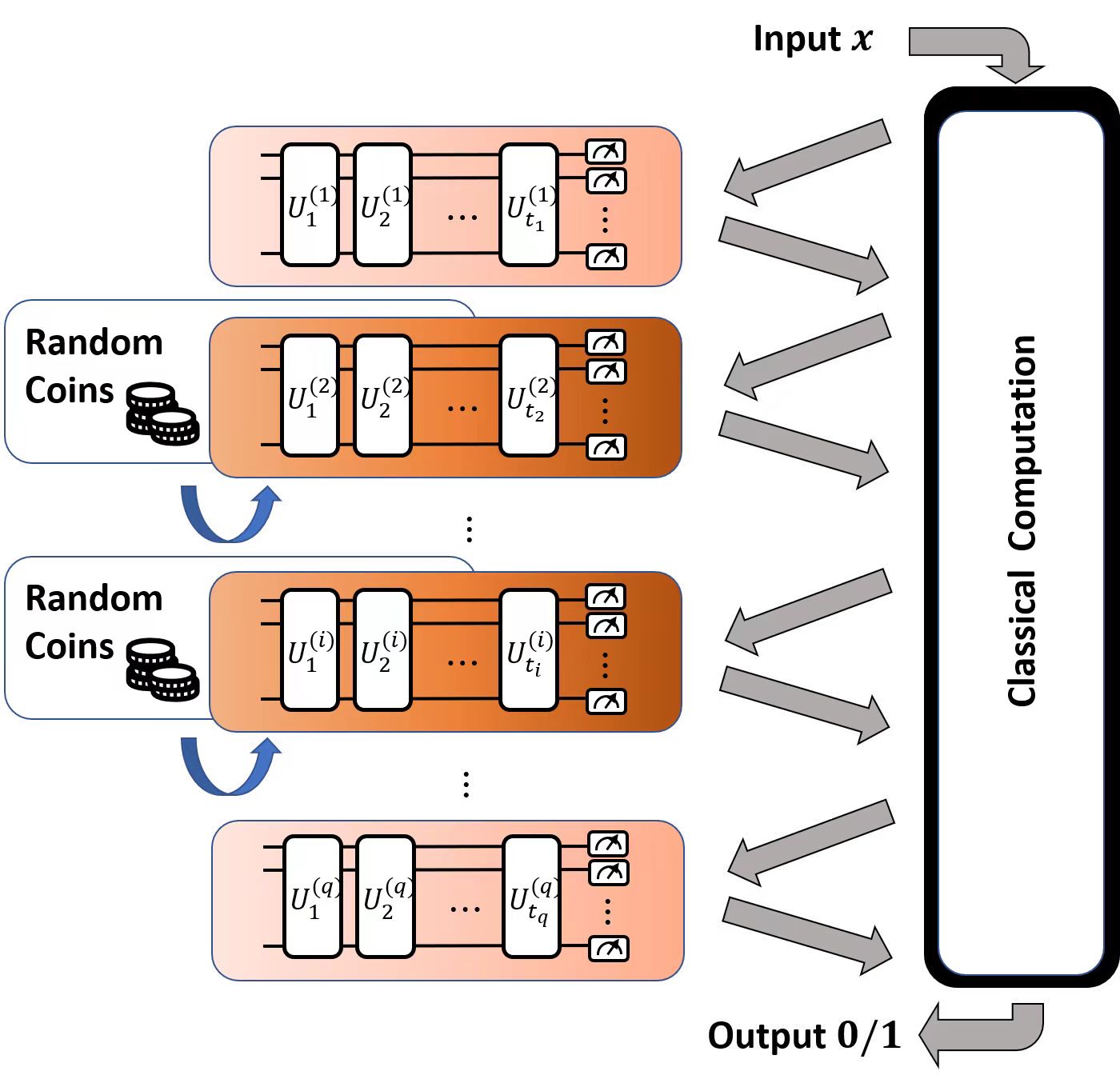}
				\label{fig:noisy-algo-coin}
			}
		\end{minipage}
		\begin{minipage}[t]{0.38\linewidth}
			\subfigure[]{
				\includegraphics[width=\linewidth]{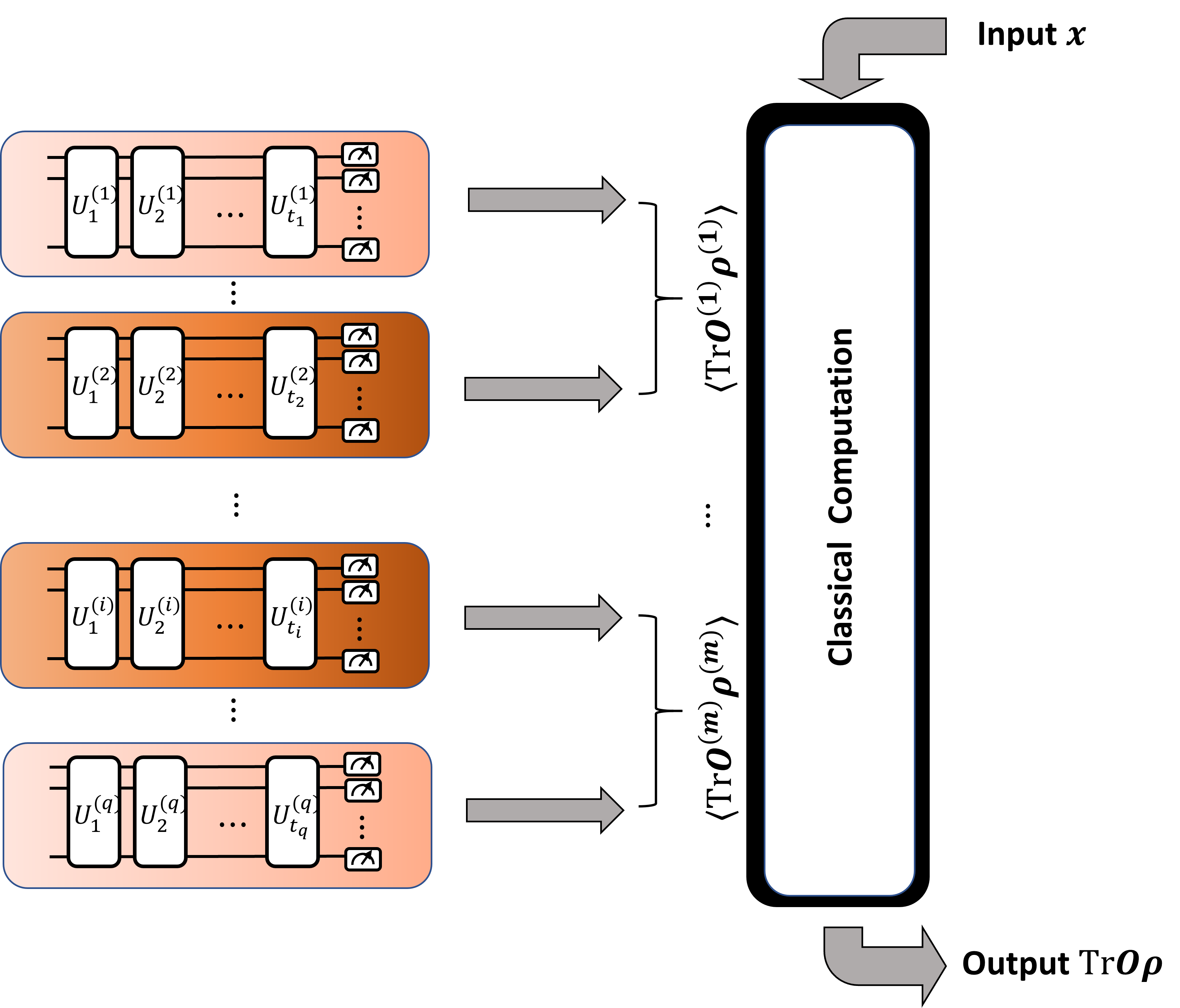}
				\label{fig:qem}
			}
			\subfigure[]{
				\includegraphics[width=\linewidth]{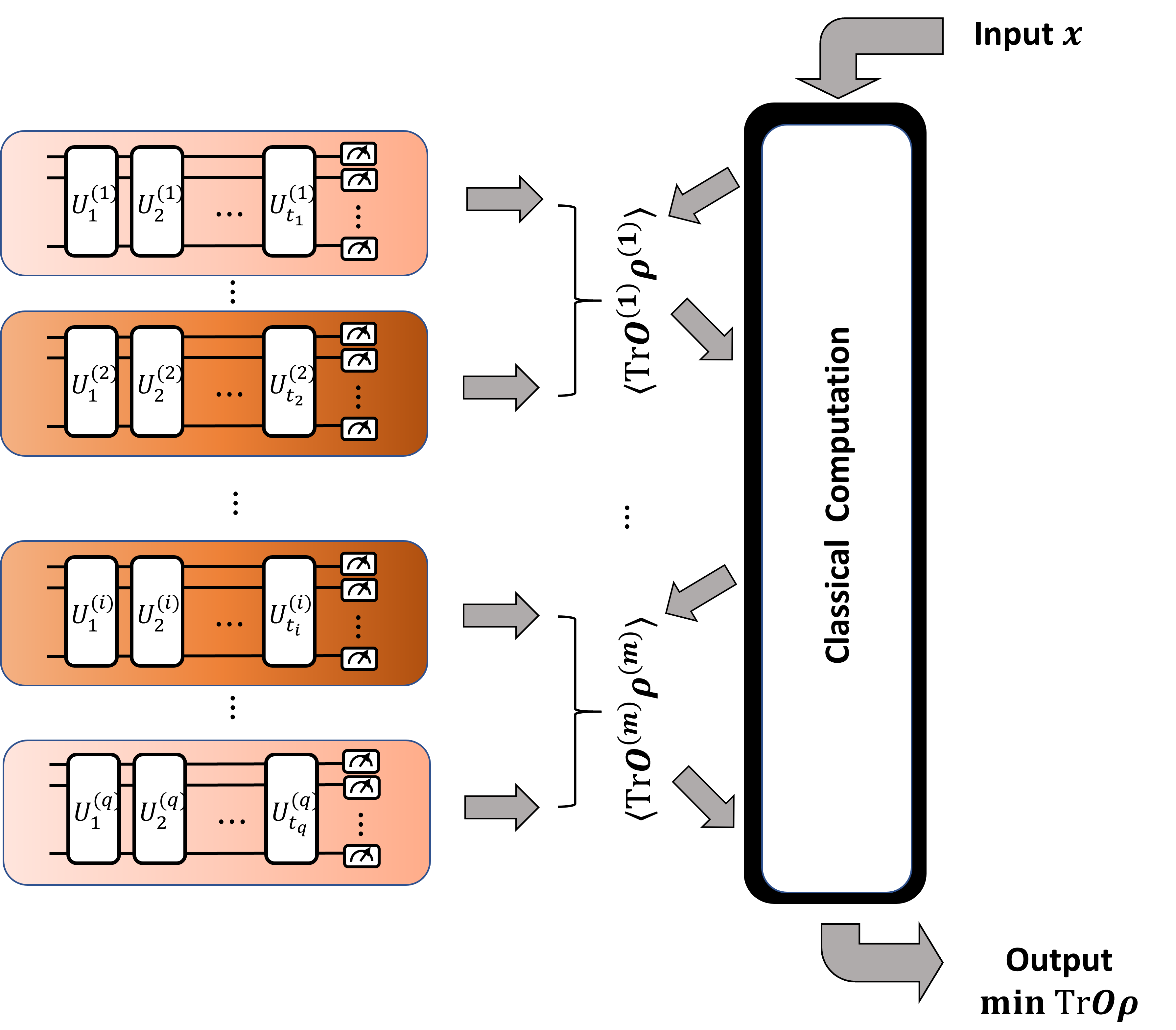}
				\label{fig:vqe}
			}
		\end{minipage}%
		\caption{\textbf{Illustration of a hybrid algorithm with noisy quantum devices and the replacement process from noisy quantum devices to random coins.} \textbf{a} Noisy quantum devices defined in Definition \ref{def:qc}. Blocks of $\mathcal{U}_i$'s are gate layers, and black dots are single-qubit noise channel $\Lambda_1$ that is assumed to be strictly contractive and unital. Computational basis measurements are performed at the end of the circuit to provide classical information output. \textbf{b} In a hybrid algorithm, the orange boxes on the left represent noisy quantum devices, in which a darker color represents noisy quantum devices with super-logarithmic depths, $\Omega(\log(nq))$. The classical computer on the right can control noisy quantum devices based on previous measurement outputs. After obtaining data from quantum devices and classical computation, the classical computer outputs 0 or 1. We use the boxes with coins to show that random coins can replace noisy quantum devices with super-logarithmic depths. After the replacement, the new classical algorithm can still give the same output bit. Thus, the noisy quantum devices with super-logarithmic depth do not provide any quantum advantages, as stated in Theorem~\ref{thm:complimit}. \textbf{c} Our results apply to single-copy quantum error mitigation, designed to suppress errors in expectation values, usually albeit with increased running time when single copies are fed. \textbf{d} Our results also apply to variational quantum algorithms that adapt circuits to update parameters during optimization. Our results suggest their limitations regarding noisy circuit depth, independent of classical processing designs.}
		\label{fig:hybrid}
	\end{figure}

	\begin{definition}[Noisy quantum devices] \label{def:qc}
		A noisy quantum device of $n$ qubits produces a quantum state at a depth of $t$,
		\begin{equation} \label{eq:rhot}
			\rho(t) = \Lambda \circ \mathcal{U}_{t} \circ \Lambda \circ \cdots \circ \Lambda \circ \mathcal{U}_{2} \circ \Lambda \circ \mathcal{U}_{1} (\ketbra{0}^{\otimes n}),
		\end{equation}
		where $\mathcal{U}_{i}$ denotes a gate layer and $\Lambda = \Lambda_1^{\otimes n}$ represents the noise channel and $\Lambda_1$ is the strictly contractive unital noise on each qubit. Each layer allows at most one gate operation per qubit. The measurement output $C_n(\rho(t))$ will then follow the distribution
		\begin{equation}
			\Pr [C_n(\rho(t)) = X] = \bra{X} \rho(t) \ket{X},
		\end{equation}
		with $X$ being an $n$-bit string and $\ket{X}$ its computational basis state.
	\end{definition}

	Our NISQ model prohibits using mid-circuit measurements and refreshing qubits. Such restrictions imply that the system's entropy cannot be reduced and will rise following each noise layer. We make the assumption intentionally for the NISQ era, as allowing these operations could facilitate fault-tolerant quantum computing below a noise threshold \cite{knill_resilient_1998, aharonov_fault-tolerant_2008}, a capability beyond the scope of the NISQ era \cite{chen_complexity_2023}.

	As a result, all stored information will inevitably be lost due to the convergence to the maximally mixed state with growing noisy circuit depth. This convergence is exponentially rapid as a function of circuit depth $t$, as reflected in the following lemma, whose proof is available in Methods.

	\begin{lemma} \label{lem:entropy}
		After $t$ layers of noisy circuits, the relative entropy between the state $\rho(t)$ and the maximally mixed state diminishes as $D(\rho(t) \| \sigma_0) \le n \mu^{t}$, where $\sigma_0 = {I}/{2^n}$ is the $n$-qubit maximally mixed state and $\mu < 1$ is the contractive rate.
	\end{lemma}

	\revise{
		Previous research has made significant findings of such contractive behavior in many noise models \cite{carbone_logarithmic_2015, kastoryano_quantum_2016, beigi_quantum_2020, bardet_modified_2021, capel_modified_2021}. Here, we extend the result to any unital contractive noise channel.

		Note that the contractive rate $\mu$ in this lemma can differ from $\mu_1$ in the single qubit case. While in special cases like depolarizing noise, we have $\mu = \mu_1$.
	}

	As highlighted above, this work focuses on quantum algorithms with the assistance of classical processing. We use the term ``hybrid algorithms'' to emphasize the integration of quantum and classical processing. Within a hybrid algorithm, a classical computer calls noisy quantum devices for measurement outcome bits. In each query, the quantum device responds by returning one bit of measurement results to the classical computer. Here, the classical computer is assumed to be noiseless with persistent memory. This process is depicted in Fig.~\ref{fig:noisy-algo-coin}. Following the intuition, strict formulations are provided in the Methods, using probabilistic Turing machines.

	Under this hybrid computing framework, we analyze the relation between the entropy of measurement results from all correlated queries $(X_1, X_2,\cdots, X_q)$ and the circuit depth $t$. Each query of a noisy quantum device produces outcome $X_i$ from computational measurements on the final gate layer, with exponentially diminishing information regarding the circuit depth $t$ \cite{aharonov_limitations_1996, kastoryano_quantum_2013}. However, extending this to multiple queries within a hybrid algorithm is not straightforward and necessitates carefully considering the correlations among the different queries. These correlations emerge from classical computers that can control subsequent quantum device operations based on previous measurements. A naive entropy analysis might suggest that such correlations could reduce the total entropy, potentially enabling hybrid algorithms to aggregate information and amplify quantum advantages.

	However, contrary to this presumption, we show that even if we consider query correlations, the aggregate information obtained from all measurement outputs $(X_1, X_2,\ldots, X_q)$ becomes exponentially small with increasing circuit depth $t$, as stated in the following lemma. The proof is made by breaking down the total entropy into a sequential sum of conditional entropies based on preceding queries, detailed in Methods.
	\begin{lemma}\label{lem:qX}
		Suppose the hybrid algorithm calls for $q$ times of circuit execution with a minimum circuit depth $t$. Denote $X_i$ as the measurement output from the $i$-th circuit execution on noisy $n$-qubit devices. The collection of outcomes $X_1, \ldots, X_q$ , containing $nq$ queried measurement bits, yields
		\begin{equation}
			S(X_1,\ldots,X_q) \ge  \big(1-\mu^{t}\big) nq,
		\end{equation}
		where $\mu$ is the contractive rate of the single-qubit noise channel.
	\end{lemma}

	This lemma considers correlations among queries to quantum devices, which arise from arbitrary classical processing and controls, thereby going beyond isolated quantum device analysis in existing studies \cite{aharonov_limitations_1996, kastoryano_quantum_2013}. The generality of this lemma makes our results applicable to generic hybrid algorithms.

	Based on the above entropy analysis, we identify the limitations of circuit depth for hybrid algorithms with noisy quantum devices. Our study reveals that the maximum depth limit to provide advantages scales as the logarithm of the number of queries to the noisy quantum devices, $nq$. We clarify this argument with the following theorem, whose formal version is available in the Methods.

	\begin{theorem}[informal version] \label{thm:complimit}
		Consider a hybrid algorithm operating within time $T(n)$ with $nq$ queried measurement bits from noisy quantum devices. If circuit depth $t \ge \Omega\left( \frac{\log (nq)}{|\log(\mu)|} \right)$, noisy devices will yield no quantum advantages. In this case, a classical algorithm, running within $O(T(n))$ time, can be used instead by replacing the noisy device queries with random coins from a uniform distribution.
	\end{theorem}

	The running time of an algorithm can be broken down into three parts: quantum circuit execution, queries to measurement outcomes, and classical processing. We can represent this as the following equation:
	\begin{equation}
		T(n) = c_1 t q + c_2 nq + T_c(n).
	\end{equation}
	Here, $t$ represents the circuit depth, $q$ represents the number of circuit executions, $T_c(n)$ represents classical processing time, and $c_1$ and $c_2$ are the constant times for executing a gate layer and querying a measurement outcome bit, respectively. We use functions of the qubit number $n$ to represent $T(n)$ and $T_c(n)$ as they typically depend on the number of qubits $n$, representing the algorithm input size. We assume that one query returns one bit of measurement outcomes, taking $c_2$ time for the classical computer. Therefore, we always have $nq \le \frac{1}{c_2}T(n)$ from the above equation.

	In previous studies on isolated noisy devices \cite{aharonov_limitations_1996, kastoryano_quantum_2016}, the maximal circuit depth depends on the number of qubits. In contrast, we highlight query times, i.e., $nq$, and the running time $T(n)$ as an important factor in hybrid quantum algorithms. The difference is attributed to the role of classical computing in hybrid algorithms. For the same noisy device, the longer the duration of a hybrid algorithm, the greater the opportunities it has to query the device and receive quantum advantages via the assistance of classical processing, e.g., error mitigation protocols \cite{temme_error_2017, endo_practical_2018, cai_quantum_2023}. Following this intuition, we present the proof in the Methods.

    The running time of efficient quantum algorithms is required to be at most polynomial growth concerning the number of qubits. In this case, where $nq \le \frac{1}{c_2} T(n) \le O(\text{poly}(n))$, noisy quantum devices with super-logarithmic depth cannot provide any advantage for such polynomial-time quantum algorithms. This limitation is unavoidable even with sophisticated classical processing or adaptive operations based on measurement outcomes in previous queries. As a result, we have established strict no-go results on algorithms with a super-logarithmic circuit depth. With noise, we cannot leverage adaptive controls to retrieve advantages of these algorithms, such as Shor's \cite{shor_polynomial-time_1997}, Grover's \cite{grover_fast_1996}, and the Harrow-Hassidim-Lloyd (HHL) algorithm \cite{harrow_quantum_2009}. In quantum simulation, our findings establish an upper bound of $O(\log(n))$ for the allowable evolution time. The reason behind this is that, generally, simulating a quantum system with an evolution time of $\tau$ requires a circuit depth proportional to $\tau$, which is known as the no-fast-forwarding theorem \cite{berry_efficient_2007,atia_fast-forwarding_2017,haah_quantum_2018}.

	Our findings apply to NISQ algorithms, where the running time and circuit depths are usually variable. In particular, we consider the algorithms that calculate the expectation values of observables, such as in variational quantum algorithms \cite{mcardle_quantum_2020, cerezo_variational_2021} and quantum error mitigation \cite{temme_error_2017, endo_practical_2018, cai_quantum_2023}. To estimate $m$ expectation values, we collect $Q_i$ computational-basis measurement outcomes for the $i$-th. Then, we use classical processing to estimate the statistical value of $f^{(i)} = \Tr [ \hat{O}^{(i)} \rho^{(i)} ]$. Here, $\hat{O}^{(i)}$ represents the $i$-th observable, and $\rho^{(i)}$ corresponds to the quantum state associated with it. Note that $\rho^{(i)}$'s may differ from each other, which are generated by various noisy circuits. Specifically, the structures of single-copy error mitigation and variational quantum algorithms are illustrated in Fig.~\ref{fig:qem} and \ref{fig:vqe}, respectively, which are special cases of our hybrid computing framework in Definition~\ref{def:qc}. Therefore, we can apply Theorem~\ref{thm:2_formal} and derive the following corollary.

	\begin{corollary}
		\label{cor:nisqalgo}
		For variational quantum algorithms and single-copy quantum error mitigation algorithms, if the circuit depth of noisy quantum devices is $t$, then the algorithm must have running time $T(n) = 2^{\Omega({|\log(\mu)| t})}$ to provide a quantum advantage.
	\end{corollary}

	The corollary unifies and strengthens important findings in NISQ algorithm limitations. For variational quantum algorithms, our results imply the issue of noise-induced barren plateaus. This issue arises at linear circuit depth due to noise, causing the gradients to exponentially vanish as the number of qubits $n$ increases and the optimization to fail \cite{wang_noise-induced_2021}. For quantum error mitigation, our results imply the exponential sampling costs associated with single-copy mitigation schemes \cite{takagi_fundamental_2022, takagi_universal_2023}. Note that a higher sampling overhead for error mitigation has been demonstrated in a fundamentally different noisy circuit model, where multiple layers of gates are executed between local depolarizing channels \cite{quek_exponentially_2023}. Hereby, we provide a unified view of previous algorithm-specific findings.

	For both aforementioned applications, we also strengthen the previous results by considering the general classical processing beyond the diagram in Fig.~\ref{fig:qem} and \ref{fig:vqe}. For example, Corollary~\ref {cor:nisqalgo} applies to hybrid enhancement to quantum variational algorithms \cite{endo_hybrid_2021, miao_neural-network-encoded_2024}, which are important algorithms but not considered in existing work \cite{stilck_franca_limitations_2021}.
	Our results also apply to learning-based error mitigation \cite{strikis_learning-based_2021}. This method uses classical processing on previous query outcomes to train a model for mitigating errors in subsequent rounds. The broader implication is that attempts to learn noisy behaviors cannot circumvent the fundamental limitations, as shown in Corollary~\ref {cor:nisqalgo}.

	\revise{
		The implications of our research further extend to sampling algorithms. Following the idea of Theorem~\ref{thm:complimit}, we show that samples from noisy devices are statistically indistinguishable from those of uniformly distributed random coins if circuit depth exceeds $\Omega(\log(nq))$. When the running time $T(n)$ is within $O(\text{poly}(n))$, it suggests that sampling advantages cannot be demonstrated in polynomial time for noisy circuit depth exceeding $\Omega(\log(n))$. The details are available in the Method section.
	}

	In real-world quantum devices, gates are often subjected to certain topologies, resulting in deeper circuits and further restraining their computational power under noise. For example, we can consider a 1D qubit array, where gates are restricted between the nearest neighbor qubits on a linear chain.

	We show that 1D noisy devices do not possess any super-polynomial computational advantages, regardless of their depth. To see this, we consider two depth regimes divided by $\log (n)$-scaling, respectively. In the $O(\log (n))$ depth regime, the pure-state dynamics of 1D chains are simulatable in polynomial time \cite{Shi2006TreeTensor, bravyi_classical_2021}. Here, we extend these results to the general cases of mixed states in noisy circuits. This generalizes previous findings to open systems, which are closer to reality. Further details are available in the Methods section.

	\begin{lemma}\label{lem:simulate_1d}
		The output distribution of 1D quantum devices, without noise or with any single-qubit noise, can be sampled with $O(2^{2t} n t)$ computational times on a classical computer, where $n$ is the qubit number, and $t$ is circuit depth.
	\end{lemma}

	In a hybrid algorithm, noisy devices are queried multiple times, as discussed in the above sections. The lemma indicates the simulatability of each single query with circuit depth $t=O(\log n)$. For the whole hybrid algorithm in polynomial running time, we can classically simulate each query and receive a polynomial overhead. Therefore, these devices cannot provide super-polynomial computational advantages in a hybrid algorithm. It is also important to note that this lemma applies to any single-qubit noise, even beyond the strictly contractive unital cases, as assumed in the above sections.

	Otherwise, if the noisy circuit is deeper, whose depth $t = \Omega(\log (n))$, Theorem~\ref{thm:complimit} suggests the absence of quantum advantages. Combining the two aspects, we exclude the super-polynomial advantages of 1D noisy quantum devices regardless of circuit depth, where the noise is strictly contractive and unital.

	\begin{theorem} \label{thm:1dcomplimit}
		One-dimensional noisy quantum devices that run in computational time $T(n)$ can be simulated by a classical algorithm with $T(n)^{\frac{2}{|\log(\mu)|} + 1}$ computational time, thus having no super-polynomial quantum advantages in a hybrid algorithm.
	\end{theorem}

	Our results stress the importance of connectivity in the NISQ diagram. The connectivity is required to be stronger than a 1D qubit array for the existence of super-polynomial advantages under noise.

	In the above sections, we have shown the computational limitations of noisy quantum devices, especially in the case of 1D connectivity. We will then delve into the connectivity of noisy devices from a more physical point of view. To this end, we analyze the entangling power in 1D and 2D qubit arrays.

	For the 1D qubit chain, we consider bipartite entanglement between two contiguous halves of the chain, denoted as $A$ and $\bar{A}$. A key observation is that the interaction of a qubit is localized in a region whose radius grows with the depth $t$. The physical picture is that the qubits interact within a light cone. We generalize this observation to entanglement spreading and derive an upper bound of the entanglement monotone $E$ between halves of the chain, $E(A:\bar{A}) \le t$.
	The generation of entanglement requires sufficient circuit depth. Similar bounds have been derived for the dynamics of pure states without noise \cite{eisert_entangling_2021, harrow_separation_2021}. Using local operation monotones and induction, we extend the entangling upper bound to the mixed-state case. This result is essential for analyzing noisy quantum devices.

	On the other hand, noisy quantum devices suffer from an exponential loss of information with increasing depth $t$, which also leads to exponentially rapid decay of maximal entanglement. Jointly, the two effects result in a logarithmic upper bound of entanglement generation in noisy quantum devices at arbitrary circuit depths, as stated in the following theorem. The proof of our results, including the following theorem, can be found in Methods.

	\begin{theorem}[Entanglement upper bound on 1D array] \label{the:1dent}
		For a contiguous half $A$ and the complement half $\bar{A}$ in an $n$-qubit noisy quantum device with a 1D connection topology, the quantum mutual information and hence the quantum relative entropy of entanglement are upper bounded by
		\begin{equation}
			E_{R}(A:\bar{A}) \le I(A : \bar{A}) \le 2 \max \left\{\frac{\log (\frac{n}{2})}{|\log (\mu)|}, 1 \right\},
		\end{equation}
		where $\mu$ is the contractive rate of the single-qubit noise channel.
	\end{theorem}

	Our findings have significant implications for preparing quantum states with large-scale entanglement by excluding the possibility of efficiently preparing any quantum state with a super-logarithmic entanglement scaling. This limitation extends to a wide range of scenarios, including high excitation \cite{huang_universal_2021} and thermalization in most quantum dynamics \cite{deutsch_thermodynamic_2010, deutsch_eigenstate_2018}. This is also related to limitations in quantum simulation, where entanglement plays an important role \cite{joshi_exploring_2023}.

	For 2D lattices, we consider qubits arranged in a square of side length $\sqrt{n}$ and show that maximal entanglement is $O(\sqrt{n} \log (n))$. For one- and two-dimensional lattices, we present the numerical upper bounds of entanglement for depolarizing at different strength levels in Fig.~\ref{fig:1dent}. After the number of qubits reaches a certain number related to the noise strength, the further growth of quantum entanglement in the noisy quantum device will be suppressed. For the 1D case, this will lead to an exponential cost of qubits required to scale up the system's entanglement further due to the logarithmic scaling of the upper bounds of entanglement. In the 2D case, a polynomial cost is also required.
	\begin{figure}[hbtp!]
		\centering
		\includegraphics[width=.7\linewidth]{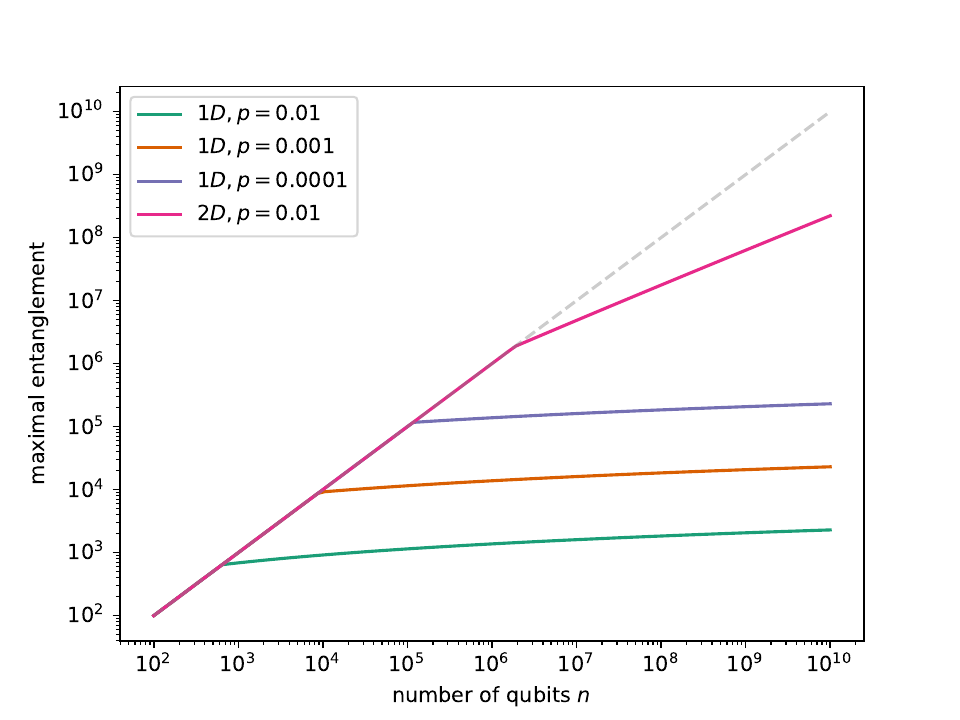}
		\caption{Limitations of entanglement generation, with varied number of qubits $n$ and single-qubit depolarizing noise strength $p$, which are taken as the minimal values between $n$ and the upper bounds in Theorem~\ref{the:1dent}. As for depolarizing noise, we have $\mu = \mu_1 = (1-p)^2$. Lines in different colors correspond to different noise strengths $p$, and one of the lines is for the case of a 2D qubit connection.} \label{fig:1dent}
	\end{figure}

	\section*{Discussion}

	The limitations on computational and entangling capabilities are inherent in our NISQ model. Future advancements in quantum computing hardware must strive to transcend the assumptions of this model and, therefore, be able to avoid the rapid convergence to the maximally mixed state. In the interim, strategies such as delaying the introduction of qubits into a noisy environment until the last possible moment or transforming the noise into more manageable forms may mitigate some of the issues \cite{wu_erasure_2022}. The ultimate solution to quantum errors can be achieved through resetting qubits or using mid-circuit measurements with feedforward actions to purify the system. These techniques are important steps toward quantum error correction, but they require a sufficient supply of qubits and low error rates below the error tolerance threshold. Therefore, exceeding the limitations outlined in our study is necessary to achieve fault-tolerant quantum computation during the current NISQ era.

	Future work will involve expanding our results to other types of noise. In this work, we focus on strictly contractive unital cases, such as depolarizing noise. However, the situation will be very different for the noise that does not always drive the state to the maximally mixed state, such as dephasing channels and the amplitude damping channel \cite{ben-or_quantum_2013}. For these noise models, as long as they are single-qubit noise, Lemma \ref{lem:simulate_1d} will still hold for any 1D qubit array, excluding superpolynomial advantages at a short depth. For higher-dimensional qubit arrays and other depth regimes, these noise models require further investigation.
	Besides single-qubit noise, NISQ devices also suffer from multi-qubit collective noise. Mathematically, the structure of a multi-qubit noise channel can be much more complicated due to the growth of dimension. We note that our results apply to multi-qubit depolarizing noise, whose contraction is similar to Lemma~\ref{lem:entropy}, as detailed in Methods. More general cases are to be further explored.

	Our results may also extend to limitations of other properties of noisy quantum devices, including quantum state complexity, topological order, magic, and quantum chaos. It will also be interesting to apply our generic results to more applications of quantum computers.

	\section*{Methods}

	First, we present the analytic expression of strictly contractive unital channels, which is given by Eq.~\eqref{eq:analytic_expression} in Proposition~\ref{prop:noise-decompose}.

	\begin{proof}[Proof of Proposition~\ref{prop:noise-decompose}]

		We first show that Eq.~\eqref{eq:analytic_expression} is a necessary condition. Consider the Bloch sphere representation of single-qubit states:
		\begin{equation}
			\rho = \frac{1}{2}\left(I + \bm{r} \cdot \bm{\sigma}\right),
		\end{equation}
		where $\bm{r} \in \mathbb{R}^3$ is a Bloch vector with $\abs{\bm{r}} \le 1$, and $\bm{\sigma} = (\sigma_X, \sigma_Y, \sigma_Z)^T$ is the vector of Pauli matrices. For any single-qubit channel $\Lambda_1$, there exist unitary operators $U$ and $V$ such that:
		\begin{equation}
			V^{\dagger} \Lambda_1\left(U^{\dagger} \rho U\right) V = \frac{1}{2}\left(I + (\bm{t} + \Sigma \bm{r}) \cdot \bm{\sigma}\right),
		\end{equation}
		where $\bm{t} \in \mathbb{R}^3$ and $\Sigma$ is a real $3 \times 3$ matrix. For unital channels, we have $\bm{t} = 0$. Moreover, $\Sigma$ can be taken as diagonal via a suitable choice of $U$ and $V$, with $\norm{\Sigma}_\infty \leq 1$ (where the norm is the operator infinity norm, i.e., the largest singular value). Thus,
		\begin{equation}
			V^{\dagger} \Lambda_1\left(U^{\dagger} \rho U\right) V = \frac{1}{2}\left(I + (\Sigma \bm{r}) \cdot \bm{\sigma}\right).
		\end{equation}
		Equivalently, we can express the channel as
		\begin{equation}
			\Lambda_1 = \mathcal{V} \circ \mathcal{D} \circ \mathcal{U},
		\end{equation}
		where $\mathcal{U}(\rho) = U \rho U^{\dagger}$ and $\mathcal{V}(\rho) = V \rho V^{\dagger}$ are unitary channels, and $\mathcal{D}$ is a diagonal channel whose Pauli-Liouville representation is given by the matrix
		\begin{equation}
			\begin{bmatrix}
				1 & \bm{0}^T  \\
				\bm{0} & \Sigma
			\end{bmatrix},
		\end{equation}
		where $\bm{0}=(0,0,0)^T$.

		We now show that strict contractivity implies $\|\Sigma\|_\infty < 1$. Recall that for a single-qubit state $\rho$ with Bloch vector $\bm{r}$, the relative entropy to the maximally mixed state is
		\begin{equation}
			D\left(\rho \| \frac{I}{2}\right) = 1 - S(\rho) = 1 - h\left(\frac{1 + \abs{\bm{r}}}{2}\right),
		\end{equation}
		where $h(x) = -x \log_2 x - (1 - x) \log_2 (1 - x)$ is the binary entropy function. Hence, $D\left(\rho \,\|\, \frac{I}{2}\right)$ is strictly increasing with $\abs{\bm{r}}$.
		Since unitary operations preserve the Bloch vector norm, the relative entropy is invariant under unitary conjugation. Now suppose, for contradiction, that $\|\Sigma\|_\infty = 1$. Then, there exists a state $\rho$ with $\bm{r} \neq 0$ such that $\Sigma \bm{r} = \bm{r}$. In this case,
		\begin{equation}
			\begin{split}
				D\left(\Lambda_1\left(U^{\dagger} \rho U\right) \,\|\, \frac{I}{2}\right)
				&= D\left(\mathcal{V} \circ \mathcal{D} \circ \mathcal{U}(U^{\dagger} \rho U) \,\|\, \frac{I}{2}\right) \\
				&= D\left(\mathcal{V} \circ \mathcal{D} (\rho) \,\|\, \frac{I}{2}\right) \\
				&= D\left(\mathcal{V}(\rho) \,\|\, \frac{I}{2}\right) \\
				&= D\left(U^{\dagger} \rho U \,\|\, \frac{I}{2}\right),
			\end{split}
		\end{equation}
		where the third line uses $\Sigma \bm{r} = \bm{r}$ (i.e., $\mathcal{D}(\rho) = \rho$), and the fourth line follows from the unitary invariance of the relative entropy. This contradicts the assumption that $\Lambda_1$ is strictly contractive, since $D$ remains unchanged. Therefore, $\|\Sigma\|_\infty < 1$.

		We now show that Eq.~\eqref{eq:analytic_expression} is also a sufficient condition for $\Lambda_1$ to be strictly contractive. Suppose $\Sigma = \mathrm{Diag}(q_X, q_Y, q_Z)$ is the diagonal matrix appearing in the Pauli-Liouville representation, where $q_X, q_Y, q_Z < 1$. Without loss of generality, assume that $q_X \le q_Y \le q_Z < 1$.
		We decompose $\Sigma$ as the product of two diagonal matrices: $\Sigma = \Sigma_1 \Sigma_2$, where
		\begin{equation}
			\Sigma_1 =  \mathrm{Diag}\left(\frac{q_X}{q_Z}, \frac{q_Y}{q_Z}, 1\right), \quad
			\Sigma_2 = \mathrm{Diag}(q_Z, q_Z, q_Z),
		\end{equation}
		The map corresponding to $\Sigma_1$ is a unital (but not necessarily completely positive) linear map $\mathcal{C}$.
		The map corresponding to $\Sigma_2$ is the single-qubit depolarizing channel $\mathcal{D}_p$ with depolarizing strength $p = 1 - q_Z$. Then, we can decompose the channel as:
		\begin{equation}
			\Lambda_1 = \mathcal{V} \circ \mathcal{C} \circ \mathcal{D}_p \circ \mathcal{U}.
		\end{equation}

		Now consider any single-qubit quantum state $\rho \neq \frac{I}{2}$. We have:
		\begin{equation}
			\begin{split}
				D\left(\Lambda_1(U^{\dagger}\rho U) \,\|\, \frac{I}{2}\right)
				&= D\left(\mathcal{V} \circ \mathcal{C} \circ \mathcal{D}_p \circ \mathcal{U}(U^{\dagger}\rho U) \,\|\, \frac{I}{2}\right) \\
				&= D\left(\mathcal{C} \circ \mathcal{D}_p (\rho) \,\|\, \frac{I}{2}\right) \\
				&\leq (1-p)^2 D\left(\mathcal{C}(\rho) \,\|\, \frac{I}{2}\right) \\
				&\le (1-p)^2 D\left(\rho \,\|\, \frac{I}{2}\right)\\
				&= (1-p)^2 D\left(U^{\dagger}\rho U \,\|\, \frac{I}{2}\right)
			\end{split}
		\end{equation}
		The second line uses the unitary invariance of relative entropy. The third line uses the contractivity of the depolarizing channel $\mathcal{D}_p$ \cite{muller-hermes_relative_2016}. The inequality in the fourth line follows because $\mathcal{C}$ does not increase the norm of the Bloch vector, and the relative entropy $D(\rho \,\|\, \frac{I}{2})$ is strictly increasing in $\abs{\bm{r}}$.

		Therefore, $D\left(\Lambda_1(\rho) \,\|\, \frac{I}{2}\right) \le (1-p)^2 D\left(\rho \,\|\, \frac{I}{2}\right)$ for all $\rho \neq \frac{I}{2}$, and thus $\Lambda_1$ is a strictly contractive unital channel. This concludes the proof.
	\end{proof}

	Based on the noise model described above, we derive the convergence to maximally mixed states in noisy quantum devices under strictly contractive unital noise.

	\begin{proof}[Proof of Lemma~\ref{lem:entropy}]
		We employ entropy contraction and tensorization techniques in Gao et al.'s previous work \cite{gao_complete_2025} and adapt them to our notation. To extend the contraction rate $\mu$ from single-qubit to multi-qubit channels, we define:
		\begin{equation}
			\mu(\Phi) \coloneqq \frac{D\big(\Phi(\rho) \,\|\, E_{\Phi}(\rho)\big)}{D\big(\rho \,\|\, E_{\Phi}(\rho)\big)},
		\end{equation}
		where $E_{\Phi}(\rho)$ is the fixed point of the channel in the limit of infinite applications of the operator $\Phi^\dagger \Phi$. That is, $E_{\Phi}(\rho) = \lim_{k\rightarrow \infty} (\Phi^\dagger \Phi)^k$. For a strictly contractive unital channel $\Lambda_1$, we have $E_{\Lambda_1}(\rho) = \frac{I}{2}$, and for the identity channel $\mathrm{id}_n$ on $n$ qubits, we have $E_{\mathrm{id}_n} = \mathrm{id}_n$.

		To upper-bound the contraction rate over larger systems, we introduce the \emph{complete entropy contraction coefficient} $\alpha_C$, defined as
		\begin{equation}
			\alpha_C(\Phi) \coloneqq \sup_{m} \mu\big(\Phi \otimes \mathrm{id}_m\big).
		\end{equation}
		Corollary 5.2 in Gao et al.'s previous work \cite{gao_complete_2025} provides an upper bound:
		\begin{equation} \label{eq:alpha}
			\alpha_C(\Lambda_1) \leq 1 - \frac{1}{2k_{\mathrm{cb}}(\Lambda_1)},
		\end{equation}
		where $k_{\mathrm{cb}}$ is the completely bounded return time, defined as
		\begin{equation}
			k_{\mathrm{cb}}(\Phi) = \inf \left\{ k \in \mathbb{N}^+ \,\middle|\, 0.9\,E_{\Phi} \leq_{\mathrm{cp}} (\Phi^\dagger \Phi)^k \leq_{\mathrm{cp}} 1.1\,E_{\Phi} \right\},
		\end{equation}
		and $\leq_{\mathrm{cp}}$ denotes ordering under complete positivity.
		Since $\Lambda_1$ is a strictly contractive unital channel acting on a finite-dimensional Hilbert space, a finite $k$ always exists such that the above inequality holds. Consequently, $k_{\mathrm{cb}}(\Lambda_1)$ is finite and $\alpha_C(\Lambda_1) < 1$ from Eq.~\eqref{eq:alpha}.

		We next use the tensorization property of $\alpha_C$, which states that for any two quantum channels $\Phi_1$ and $\Phi_2$, we have
		\begin{equation}
			\alpha_C(\Phi_1 \otimes \Phi_2) = \max \big\{ \alpha_C(\Phi_1), \alpha_C(\Phi_2) \big\}.
		\end{equation}
		Applying this recursively, we obtain that the complete entropy contraction coefficient of the $n$-qubit noise channel $\Lambda = \Lambda_1^{\otimes n}$ is
		\begin{equation}
			\alpha_C(\Lambda) = \alpha_C(\Lambda_1) < 1.
		\end{equation}

		By definition, the contraction coefficient $\mu$ satisfies $\mu \leq \alpha_C(\Lambda) < 1$. Thus, for any quantum state $\rho$,
		\begin{equation}
			D\big(\Lambda(\rho) \| \sigma_0\big) \leq \mu  D\big(\rho \| \sigma_0\big),
		\end{equation}
		where $\sigma_0 = \frac{I}{2^n}$ is the maximally mixed state on $n$ qubits.
		Furthermore, unitary operations in each layer preserve the relative entropy due to its unitary invariance. Hence, they do not affect the contraction.
		Iterating this inequality over $t$ layers of the noisy circuit yields the desired convergence bound stated in Lemma~\ref{lem:entropy}.
	\end{proof}

	To strictly characterize the computational capacity, we propose the formalism of hybrid quantum algorithms, which combine quantum computing with classical processing, as illustrated in Fig.~\ref{fig:noisy-algo-coin}. The scope of the algorithms we discuss includes decision problems and sampling problems. We adopt the notion of languages and the probabilistic Turing machine (PTM), which are standard terms in theoretical computer science. For decision problems, the term \textit{language} refers to a general problem where a determined answer is required as an answer to the problem.
	\begin{definition}[Language]\label{def:lg}
		A language $L$ is a subset of $\{0,1\}^{*}$. For $x \in \{0,1\}^*$, $L(x)$ is defined as $L(x) = [x \in L]$.
	\end{definition}

	A probabilistic Turing machine (PTM) is a classical algorithm that can compute a language with a probability of giving the correct answer greater than $\frac{2}{3}$. It is important to note that the classical algorithm may be randomized. Regarding the randomness in algorithms, a random variable $Y$ is introduced to represent the choice of Turing machines.

	\begin{definition}[Probabilistic Turing machine] \label{def:ptm}
		A probabilistic Turing machine, denoted by $M$, decides a language $L$ in time $T(n)$ if, for any string $x$, $M$ halts within $T(|x|)$ steps and the probability of $M$ outputting the correct answer for $x$ is at least $\frac{2}{3}$ when given a random string $Y$.

		Here, $Y$ is a random choice of Turing machines, which follows a uniform distribution over bit strings of length $T(|x|)$. When the input $x$ and random choices $Y$ are given, $M(x, Y)$ is the output of the chosen Turing machine.
	\end{definition}

	Hybrid algorithms can interact with noisy quantum devices, as shown in Fig.~\ref{fig:noisy-algo-coin}.  We formulate a hybrid algorithm as a $\PTM$ that queries noisy quantum devices, receives measurement outcomes in the form of bit strings, and performs classical processing on these outcomes. As mentioned in the main text, $T(n)$ denotes its running time, $t$ denotes its circuit depth, and $q$ denotes the circuit execution times.

	Then, we provide proof of Lemma \ref{lem:qX} with a slight generalization. We recall the definition of conditional entropy.
	\begin{definition}[Conditional Entropy]
		For two discrete random variables $X$ and $Y$, the conditional entropy is defined as
		\begin{equation}
			S(X|Y) = -\sum_{x,y} P(X=x,Y=y) \log P(X=x|Y=y).
		\end{equation}
	\end{definition}\label{def:cond_entropy}

	The conditional entropy is useful in our proof. We briefly prove two entropy equalities with direct calculation for clarity and completeness. Firstly,
	\begin{equation} \label{eq:cond_entropy}
		\begin{split}
			S(X|Y)& = - \sum_y P(Y=y) \sum_x P(X=x|Y=y) \log P(X=x|Y=y) \\
			&= \sum_y P(Y=y) S(X|Y=y),
		\end{split}
	\end{equation}
	in which $S(X|Y=y)$ is the entropy of marginal distribution of $(X,Y)$ when $Y$ is fixed value $y$. Secondly,
	\begin{equation}\label{eq:cond_entropy2}
		\begin{split}
			S(XY) &= -\sum_{x,y} P(X=x,Y=y) \log P(X=x,Y=y) \\
			&=  -\sum_{x,y} P(X=x,Y=y) [\log P(X=x|Y=y) + \log P(Y=y)] \\
			&=  -\sum_{x,y} P(X=x,Y=y) \log P(Y=y) - \sum_{x,y} P(X=x,Y=y) \log P(X=x|Y=y) \\
			&=  -\sum_{y} P(Y=y) \log P(Y=y) - \sum_{x,y} P(X=x,Y=y) \log P(X=x|Y=y) \\
			&= S(Y) + S(X|Y).
		\end{split}
	\end{equation}

	After introducing the notion of conditional entropy, we provide proof of Lemma \ref{lem:qX}.
	Here, we consider a slightly more general case, where the hybrid algorithm can query quantum devices with different numbers of qubits $n_i$.
	\begin{lemma}[A generalized version of Lemma \ref{lem:qX}] \label{lem:qXapp}
		Suppose we use noisy quantum devices of depth at least $t$ for $q$ times. Let $X_i = C_{n_i}(\rho_i)$ denotes the measurement result of the $i$-th quantum circuit. Here, $\rho_i = \Phi^{(i)}(\ket{0}\bra{0}^{\otimes n_i})$, $\Phi^{(i)} = \Lambda \circ \mathcal{U}_{t_i}^{(i)} \circ \Lambda \circ \cdots \circ \Lambda \circ \mathcal{U}_{2}^{(i)} \circ \Lambda \circ \mathcal{U}_{1}^{(i)}$ denotes the quantum channel as a whole process combining all gates and noise in sequential order, $t_i \ge t$, $\mathcal{U}_j^{(i)}$ is an arbitrary quantum channel. We obtain $q$ random variables $X_1, \dots, X_q$. Then $S(X_1,\cdots,X_q) \ge (\sum_{i=1}^q n_i) (1-\zeta)$, in which $\zeta = \mu^{t}$.
	\end{lemma}
	\begin{proof}
		By Lemma \ref{lem:entropy}, for any quantum channel $\mathcal{U}^{(i)}_j$,
		\begin{equation}
			S(\Phi^{(i)}(\ket{0}\bra{0}^{\otimes n_i})) \ge n_i (1- \zeta),
		\end{equation}
		with $\zeta = \mu^{t}$.

		According to our assumption, measurements are noiseless and performed on a computational basis. The resulting distribution of measurement outcome $X_i$ is on the diagonal of $\Delta(\Phi^{(i)}(\rho))$, where $\Delta$ denotes the dephasing channel for the computational basis. The entropy of the dephased state is the Shannon entropy of random strings $X_i$,
		\begin{equation}
			S(X_i) = S(\Delta(\Phi_t(\rho))).
		\end{equation}

		Note that $S(\Delta(\rho)) \ge S(\rho)$ for an arbitrary quantum state $\rho$ from the quantum data processing inequality. Thus,
		\begin{equation}\label{eq:ent}
			S(X_i) = S(\Delta(\Phi_t(\rho))) \ge S(\Phi^{(i)}(\rho)) \ge n_i (1- \zeta).
		\end{equation}

		Note that equation Eq.~\eqref{eq:ent} holds regardless of channels $\{\mathcal{U}^{(i)}_j\}_{j = 1,2,\cdots,t_i}$ implemented in the noisy quantum device. Although $X_i$ may depend on $X_1, \dots, X_{i-1}$, Eq.~\eqref{eq:ent} holds regardless of previous measurement outcome. In other words, for all $1 \le i \le q$ , $1 \le j < i$ and $x_j \in \{0,1\}^{n_j}$, we have
		\begin{equation}
			S(X_i|X_1=x_1,X_2=x_2,\cdots,X_{i-1}=x_{i-1}) \ge n_i (1-\zeta).
		\end{equation}

		Then we sum over all previous measurement outcomes $x_1,\cdots,x_{i-1}$ and use Eq.~\eqref{eq:cond_entropy},
		\begin{equation}\label{eq:total_entropy}
			\begin{split}
				S(X_i|X_1,X_2,\cdots,X_{i-1}) &= \sum_{x_1,\cdots,x_{i-1}}P(X_1=x_1,X_2=x_2,\cdots,X_{i-1}=x_{i-1}) S(X|X_1=x_1,X_2=x_2,\cdots,X_{i-1}=x_{i-1})\\
				& \ge \sum_{x_1,\cdots,x_{i-1}}P(X_1=x_1,X_2=x_2,\cdots,X_{i-1}=x_{i-1}) n_i(1-\zeta)   \\
				& \ge n_i(1-\zeta).
			\end{split}
		\end{equation}
		By Eq.~\eqref{eq:cond_entropy2},
		\begin{equation}
			\begin{split}
				S(X_1,X_2,\cdots,X_q) &= S(X_q|X_1,X_2,\cdots,X_{q-1}) + S(X_1,X_2,\cdots,X_{q-1}) \\
				&= S(X_q|X_1,X_2,\cdots,X_{q-1}) + S(X_{q-1}|X_1,X_2,\cdots,X_{q-2})  +  S(X_1,X_2,\cdots,X_{q-2})\\
				&= \sum_{i=1}^q S(X_i | X_1,X_2\cdots,X_{i-1}).
			\end{split}
		\end{equation}
		Combined with \eqref{eq:total_entropy},
		\begin{equation}
			S(X_1,X_2,\cdots,X_q)\ge \sum_{i=1}^q n_i(1-\zeta).
		\end{equation}
	\end{proof}

	Considering the case where all devices have the same number of qubits $n_i = n$ and the number of queried measurement bits $Q = \sum_{i=1}^{q} n_i = qn$, we will immediately obtain Lemma \ref{lem:qX}.

	Based on the above lemma, we analyze the limitations of hybrid algorithms with noisy quantum devices, leading to Theorem~\ref{thm:complimit}, as formally stated below.

	\begin{manualtheorem}{1}[formal version] \label{thm:2_formal}
		Consider a hybrid algorithm $\cA$ that decides a decision problem $L$ with $nq$ queried measurement outcome bits in running time $T(n)$. A classical algorithm $M$ exists that decides $L$ in time $O(T(n))$ if $t \geq \frac{1}{|\log(\mu)|} (\log(nq) + 5)$, where $\mu$ is the contractive rate of the unital, strictly contractive noise channel. Here, $M$ can be constructed by replacing noisy quantum devices with random coins from a uniform distribution.
	\end{manualtheorem}

	\begin{proof}
		For convenience, denote the depth upper limit in Theorem~\ref{thm:complimit} by $t^\star = \frac{1}{ |\log(\mu)|} (\log(Q) + 5)$. As stated in the theorem, we will consider the case where $t \ge t^\star$.

		Let $Y$ denote the random strings in $\cA$, including the string from noisy quantum devices and random inputs for the probabilistic Turing machines. Let $W_1$ denote $Y$'s substring from noisy quantum devices, namely, $W_1 = Y_{\mathcal{S}}$, where $\mathcal{S}$ is the indices that correspond to noisy quantum devices queries.

		Now, consider a construction of a classical algorithm, i.e., PTM $M$. Replace $W_1$ in $Y$ with a random string $W_2$ from a uniform distribution, and denote the new string as $Z$, such that $W_2 = Z_{\mathcal{S}}$.
		\revise{
		By Lemma \ref{lem:qX},
		\begin{equation}\label{eq:W1W2}
			D(W_1\|W_2)  \le  Q \zeta,
		\end{equation}
		in which $\zeta = \mu^{t}$.
		}
		Here, $Y, Z$ can be regarded as generated by passing $W_1$, $W_2$ through the same channel $\Gamma$, that is, $Y = \Gamma(W_1), Z = \Gamma(W_2)$. The channel $\Gamma$ represents the transfer from quantum measurement outcomes to the rest of classical random strings, determined by classical processing and controls. Then, by the data processing inequality, we have
		\begin{equation} \label{eq:D_yz}
			D(Y\|Z) = D(\Gamma(W_1)\|\Gamma(W_2)) \le D(W_1\|W_2) \le Q \zeta \le \frac{1}{32}.
		\end{equation}
		The first inequality is the data processing inequality; the second is Eq.~\eqref{eq:W1W2}; the third is from the assumption $t \geq t^\star$. By construction, $\cA$ and $M$ has same classical processing structure, that is, $\forall x \in \{0,1\}^n,y \in \{0,1\}^{T(n)}, M(x,y) = M'(x,y)$. Then we have the following equation:
		\begin{equation}
			D(\cA(x,Y)\|M(x,Z)) \le D(Y\|Z) \le \frac{1}{32}.
		\end{equation}
		The first inequality is from the data processing inequality; the second is from Eq.~\eqref{eq:D_yz}. Pinsker's inequality suggests
		\begin{equation}
			\|\cA(x,Y)-M(x,Z)\|_1 \le \sqrt{2D(\cA(x,Y)\|M(x,Z))} \le \frac{1}{4}.
		\end{equation}

		We have the probability that the classical algorithm solves the decision problem $L$,
		\begin{equation}\label{eq:replace_prob}
			\Pr[M(x,Z) = L(x)] = \Pr[M(x,Z) = \cA(x,Y)] \ge \frac{2}{3} - \frac{1}{8} = \frac{13}{24}.
		\end{equation}
		Then, we can repeat the $\PTM$ $M$ a constant number of times and then take a majority vote on the output results to obtain the final decision. By doing so, we can ensure that the probability of the output being equal to $L(x)$ is at least $\frac{2}{3}$. Therefore, $L$ can be decided by a $\PTM$ without noisy quantum devices in $O(T(n))$ time.
	\end{proof}

	This leads to the absence of quantum advantages of many existing algorithms, including the examples mentioned in the main text, such as Shor's, Grover's, and HHL algorithms. For Shor's and HHL algorithms, the algorithm's output is not one bit as assumed in our decision problem formalism. In the next section, we will show that they can still be reduced to a decision problem, therefore, within the scope of Theorem~\ref{thm:complimit}.

	Furthermore, we extend our results to quantum sampling, which has been demonstrated in experiments \cite{arute_quantum_2019, wu_strong_2021}. In sampling problems, the quest is to obtain $nq$ measurement output bits, i.e., samples, from noisy quantum devices. Here, we consider a distinguisher that tries to tell whether samples are from noisy quantum devices or a uniform distribution. Such a task is a decision problem; therefore, according to Definition~\ref{def:ptm}, we require the distinguisher to succeed with probability at least $\frac{2}{3}$. We show that such a distinguisher will fail if circuit depths exceed the same depth upper limits in Theorem~\ref{thm:2_formal}.

	\begin{theorem} \label{thm:sample}
		Consider $nq$ samples generated from noisy quantum devices. If circuit depth $t \geq \frac{1}{|\log(\mu)|} (\log(nq) + 5)$, then the obtained samples are statistically indistinguishable from those from a uniform distribution. Namely, any distinguisher cannot tell the difference with probability at least $\frac{2}{3}$.
	\end{theorem}

	\begin{proof}
		Consider the distinguisher, mentioned in the theorem, as a hybrid algorithm $\cA$ with no input $x$ and only queries to noisy quantum devices. We require $\cA(\emptyset, Y)$ to tell whether it actually queries noisy quantum devices or just random coins from a uniform distribution. Specifically, when given quantum devices, it should return a ``true'' with probability at least $\frac{2}{3}$; when given random coins, it should always return a ``false.''

		Suppose it can indeed solve this decision problem. Now, we replace noisy quantum devices with random coins, equivalent to $M(\emptyset, Z)$. From Eq.~\eqref{eq:replace_prob}, with probability at least $\frac{13}{24}$, the distinguisher will return the same answer in the two cases, $\cA(\emptyset, Y)$ and $M(\emptyset, Z)$. Note that when given random coins, it should return a ``false.'' Therefore, when given noisy quantum devices, it will answer ``false'' with probability at least $\frac{13}{24}$, which is a wrong answer. Therefore, it cannot succeed with at least $\frac{2}{3}$ probability, i.e., failing the decision problem.
	\end{proof}

	In what follows, we present the equivalence between some typical sorts of problems and a decision problem. Here, equivalence means that solving one problem in polynomial time implies being able to solve the other problem in polynomial time.

	\begin{proposition}[Factorizing]\label{prop:Shor}
		The following two problems are equivalent:
		\begin{enumerate}
			\item {Given $n$, output the smallest non-trivial factor of $n$.}
			\item {Given $(n,k)$, determine if there exist a non-trivial factor of $n$ that is less than $k$.}
		\end{enumerate}
	\end{proposition}

	\begin{proof}
		Suppose we solve problem 1 by algorithm $\cA$ in polynomial time. Then, for any input $(n,k)$, we use $\cA$ to factorize $n$ and get its minimal non-trivial factor $m$. Then we compare $m$ and $k$. Thus, we solve problem 2 in polynomial time.

		Suppose we solve problem 2 by algorithm $\cB$ in polynomial time. Then, we can use binary search to find the smallest non-trivial factor of $n$ in polynomial time.
	\end{proof}

	\begin{proposition}[Linear systems of equations]\label{prop:HHL}
		The following two problems are equivalent:
		\begin{enumerate}
			\item {Given a classical description of the $N \times N$ matrix $A$, a unit vector $\ket{b}$, a quantum operator $M$, and a precision $\epsilon$. Output $\langle x | M | x \rangle$  with precision $\epsilon$, $\ket{x}$ is the solution of $A\ket{x} = \ket{b}$.}
			\item {Given a classical description of the $N \times N$ matrix $A$, a unit vector $\ket{b}$, a quantum operator $M$, a real number $a$ , and a precision $\epsilon$. Determine if $\langle x | M | x \rangle$ is less than $a + \epsilon$ (output 0), or is larger than $a + \epsilon / 2$(output 1). $\ket{x}$ is the solution of $A\ket{x} = \ket{b}$.}
		\end{enumerate}
	\end{proposition}

	\begin{proof}
		Suppose we solve problem 1 by algorithm $\cA$ in polynomial time. For an input $(A,\ket{b}, M, a, \epsilon)$ of problem 2, we query a with input $(A,\ket{b}, M, a, \epsilon / 10)$ to get an output $a'$. If $a' < a + \frac{9\epsilon}{10}$, this means  $\bra{x}M\ket{x} \le a' + \frac{\epsilon}{10} < a + \epsilon$, then output 0. Otherwise, $\bra{x}M\ket{x} \ge a' - \frac{\epsilon}{10} \ge a + \frac{4\epsilon}{5}$ output 1. Then, we solve problem 2 in polynomial time.

		Suppose we solve problem 2 by algorithm $\cB$ in polynomial time. For an $(A,\ket{b}, M, \epsilon)$ of problem 1, we can solve problem 1 by performing a binary search in $O( \log (\frac{1}{\epsilon}))$ time on algorithm $\cB$ with decision $\epsilon' = \frac{\epsilon}{10}$ inputted to $\cB$.
	\end{proof}

	Proposition \ref{prop:Shor} and Proposition \ref{prop:HHL} suggest the applicability of Theorem~\ref{thm:complimit} to problems of factoring and solving linear systems, corresponding to Shor's and HHL algorithms, respectively.

		With gate topology constraints, the computational capacity is further weakened, as shown in Lemma~\ref{lem:simulate_1d}. To prove the lemma, we propose an efficient classical simulation algorithm of 1D noisy circuits, detailed in Algorithm~\ref{alg:1d}. To start with, we introduce a decomposition of the circuit, illustrated in FIG.~\ref{fig:LV}.

		\begin{figure}[hbtp!]
			\centering
			\begin{minipage}{0.48\textwidth}
				\subfigure[]{
					\includegraphics[scale=0.3]{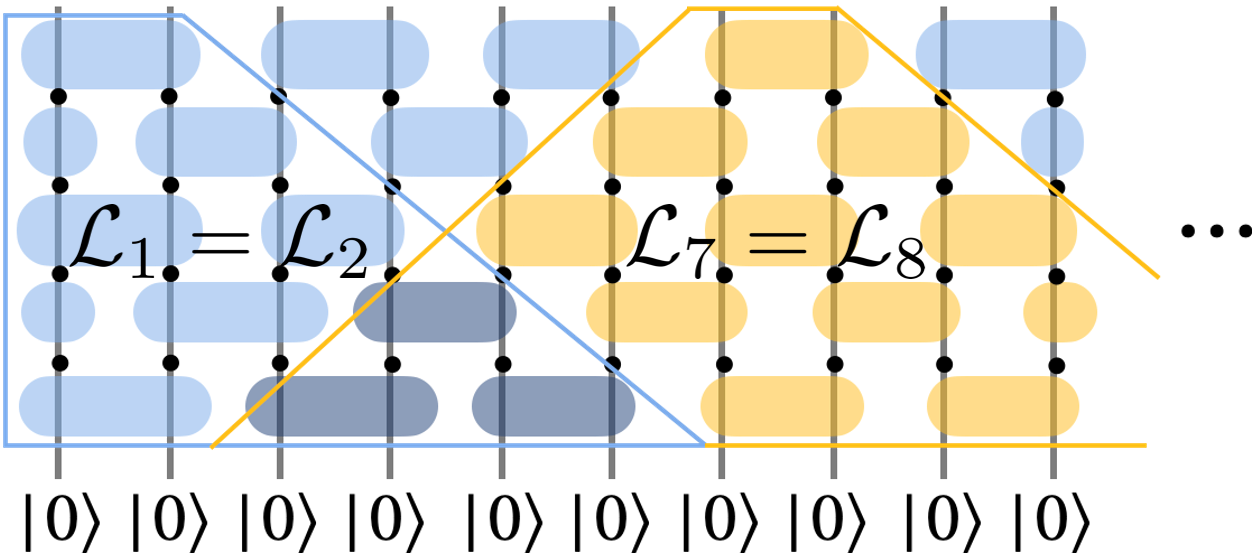}
					\label{fig:L}
				}
			\end{minipage}
			\begin{minipage}{0.48\textwidth}
				\subfigure[]{
					\includegraphics[scale=0.3]{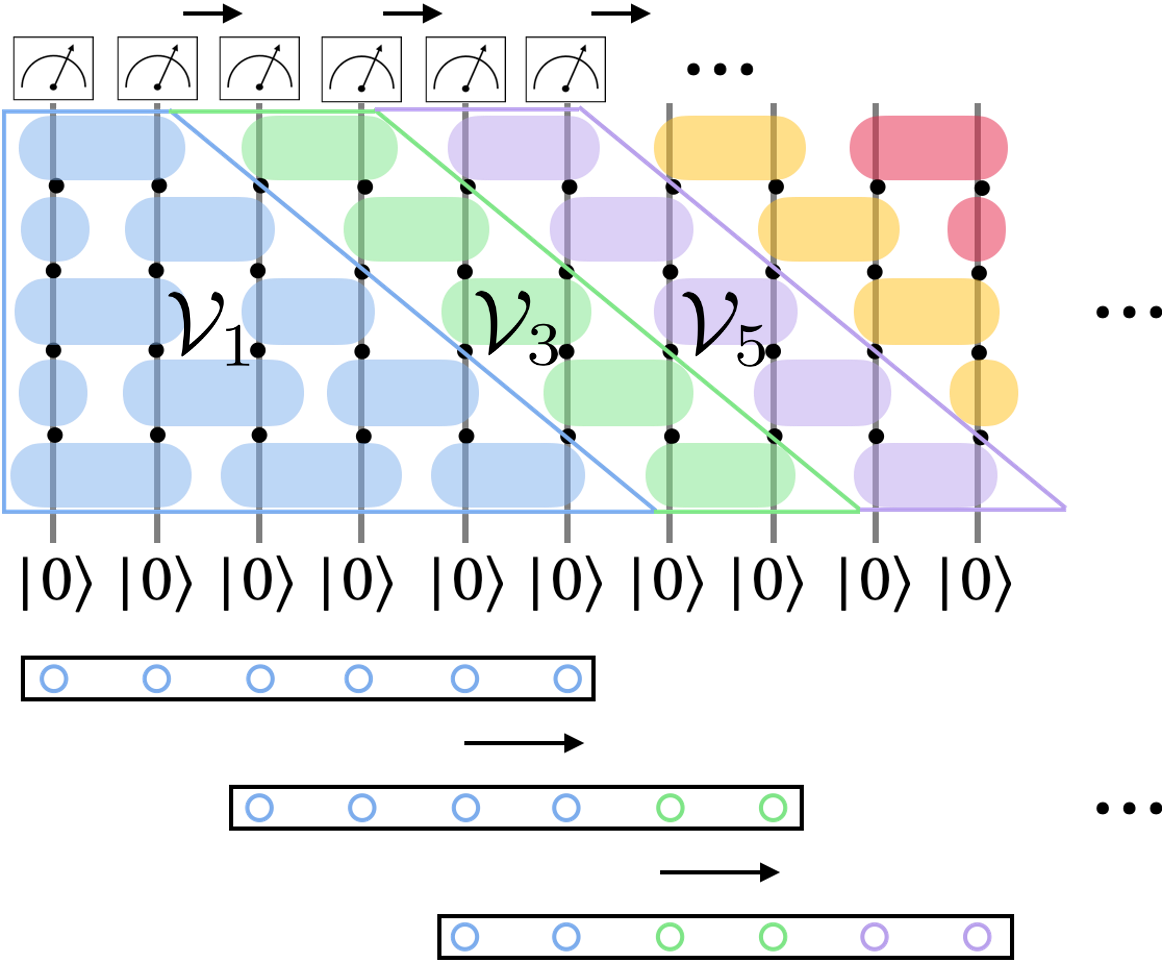}
					\label{fig:V}
				}
			\end{minipage}

			\caption{\textbf{Illustration of the light cone and effective channels for each site and the sampling algorithm.} Here, we show an example of brickwise circuits. \textbf{a} For each site, only quantum gates and noise channels within the light cone of the $i$-th qubit impact the expectation value of local observables at site $i$. The number of qubits in each light cone is at most $2t$. \textbf{b} Due to the limited range of light cones, we can define an effective channel $\cV_i$ on each site and the dynamic sweeping process in the algorithm. In the classical computer, we store the density matrix $t + 1$ qubits, where $t$ represents circuit depth. The procedure involves sweeping over each site to perform a measurement and subsequently discarding the measured qubit. Following this, we introduce the next qubit into the classical computer, applying effective channels on the classical memory as per Eq.~\eqref{eq:sim}. This iterative process continues as we alternatively calculate the conditional probability on each site and sample accordingly. After finishing the whole process, we will obtain a sample string $\bx$ from $p(\bx) = |\bra{x}U\ket{0}^{\otimes n}|^2$. The algorithm operates with a time complexity of $2^{O(t)}$.}
			\label{fig:LV}
		\end{figure}

		The 1D noisy circuit of depth $t$ is a sequence of unitary gates and noise channels, denoted as $\cU$,
		\begin{equation} \label{eq:circuitchannel}
			\cU = \Lambda \circ \mathcal{U}_{t} \circ \Lambda \circ \cdots \circ \Lambda \circ \mathcal{U}_{2} \circ \Lambda \circ \mathcal{U}_{1},
		\end{equation}
		where $\Lambda$ and $\cU_{i}$ are noise channel and unitary gate layers appeared in Eq.~\eqref{eq:rhot}, respectively. We give the following definition of the light cone of each qubit, which is illustrated in FIG.~\ref{fig:L}.

		\begin{definition}[Light cone] \label{def:lightcone}
			For a 1D noisy circuit $\cU$ of depth $t$, the light cone of the $i$-th qubit is defined as the set of all gates and noise channels that can affect the expectation value of local observables at site $i$. The light cone is denoted as $\cL_i$, which includes all gates and noise channels within a distance of $t$ from the $i$-th qubit.
		\end{definition}

		By construction of the light cones, for any observable $O_i$ and density matrix $\rho$, we have $\Tr[\cU(\rho) O_i] = \Tr[\cL_i(\rho) O_i]$. Note that $\cL_i$ only acts non-trivially on $O(t)$ qubits.

		Based on the light cones, we construct the effective channel for each site, denoted as $\cV_i$, following these steps: Let $\cV_1 = \cL_1$. For $i > 1$, define $\cV_i$ by removing the overlapped gates and noisy channels in $\cL_i$ and $\cL_{i-1}$ from $\cL_i$, as illustrated in Fig.~\ref{fig:V}. This leads to a decomposition of the circuit channel in Eq.~\eqref{eq:circuitchannel},
		\begin{equation}
			\begin{split}
				\cU = \cV_n \circ \cV_{n-1} \circ \cdots \circ \cV_1. \label{eq:1ddecomp}
			\end{split}
		\end{equation}
	Based on the decomposition, we introduce the following algorithm. In this algorithm, we classically store the density matrix of $t+1$ qubits. The stored qubits are denoted by the set $T$, which are updated during iterations, illustrated by FIG.~\ref{fig:V}.

	\begin{algorithm}
			\caption{Classical algorithm to sample from 1D noisy quantum circuits} \label{alg:1d}
			\KwData{Circuit $\cU$ in the decomposed form Eq.~\eqref{eq:1ddecomp}.}
			\KwResult{Sample $\{ a_1, \cdots, a_n \}$.}
			$T = \{1, \cdots, t+1\}$\;
			$\rho_{T} \gets \ketbra{0}{0}^{\otimes t+1} $\;
			$i \gets 1$\;
			\For{$i = 1, \cdots, n$}{
				$\rho \gets \cV_i(\rho_{T})$\;
				$a_i \gets \text{sample}\left(p(x_i = 0 | x_1=a_1, \cdots, x_{i-1}=a_{i-1}) = \Tr[\ketbra{0}{0}_i \rho]\right)$\;
				$\rho \gets \frac{\ketbra{a_i}{a_i}_i \rho \ketbra{a_i}{a_i}_i}{\Tr[\ketbra{a_i}{a_i}_i \rho]}$\;
				\If{$i < n - t$}{
					$T \gets T / \{i\} \cup \{i+t+1\}$\;
					$\rho_{T} \gets (\Tr_{i} \rho) \otimes \ketbra{0}{0}_{i+t+1}$\;
				}
			}
	\end{algorithm}

	Algorithm~\ref{alg:1d} follows the method proposed in \cite{bravyi_classical_2021} and adapts it to noisy 1D circuits. Note that for our 1D case, we do not need matrix-product-state approximation as in \cite{bravyi_classical_2021} and the algorithm is therefore exact. In the following proposition, we will show that Algorithm~\ref{alg:1d} provides an exact sampling of the noisy circuits with an analysis of the space and time complexity.

	\begin{proposition}
		Algorithm~\ref{alg:1d} provides an exact sampling of the noisy 1D circuits, which runs in $O(2^{2t} nt)$ computational time with space complexity $O(2^{2t})$, where $t$ is the depth of the circuit and $n$ is the qubit number.
	\end{proposition}

	\begin{proof}
		To sample from each qubit, we need to consider all gates and noise within its light cone, defined in Definition \ref{def:lightcone}. Other parts of the circuit will not affect the sampling result.

		Further, based on the notion of effective channels, we perform the sampling procedure for $p(\bx)$ by successively sampling from the conditional distribution $p(x_i | x_{1}=a_1, \ldots, x_{i-1}=a_{i-1})$. For convenience, we define $[n] = \{1,2,\cdots,n\}$, $M_j = \ketbra{a_j}{a_j} \otimes I_{[n] - j}$ for $a_j \in \{0,1\}$. When measuring the $j$-th qubit of $\rho$, the probability of obtaining result $a_j$ is $\Tr[M_j \rho]$, and the post-measurement state is $S_j(\rho) = \frac{M_j \rho M_j}{\Tr[M_j \rho]}$. With previous measurement outcomes fixed, the conditional probability can be written as
		\begin{equation}
			\begin{split}
				&p(x_i = 0|x_{1}=a_1,\cdots,x_{i-1}=a_{i-1}) \\
				= &\Tr[\ketbra{0}{0}_i S_{i-1} \circ S_{i-2} \circ \cdots \circ S_{1} \circ \cU (\rho_0)]\\
				= &\Tr[\ketbra{0}{0}_i S_{i-1} \circ S_{i-2} \circ \cdots S_{1} \circ \cV_n \circ \cV_{n-1} \circ \cdots \circ \cV_{1} (\rho_0)],
			\end{split}
		\end{equation}
		where $\rho_0 = \ketbra{0}{0}^{\otimes n}$ is the initial state. We observe that $S_i$ and $\cV_{i+1}$ commute since they act on different qubits. Therefore, we rearrange the operators in the equation above so that $\cV_i$ and $S_i$ are applied alternatively,
		\begin{equation}\label{eq:sim}
			\begin{split}
				&p(x_i = 0|x_{1}=a_1,\cdots,x_{i-1}=a_{i-1}) \\
				=& \Tr[\ketbra{0}{0}_i \cV_i \circ S_{i-1} \circ \cV_{i-1} \circ S_{i-2} \circ \cV_{i-2} \circ\cdots S_{1} \circ \cV_1(\rho_0)].
			\end{split}
		\end{equation}

		The conditional probability in Eq.~\eqref{eq:sim} can be estimated by simulating $O(t)$ qubits in a classical computer. Specifically, $\cV_1(\rho_0)$ is supported on $(t + 1)$ qubits, necessitating $2^{O(t)}$ time to compute the density matrix. As we have a complete description of $\cV_1(\rho_0)$ in the classical computer, we can calculate both $p(x_1=a_1)$ and the post-measurement state $S_1(\cV(\rho_0))$. The subsequent measurement process for the second qubit follows a similar procedure. Subsequently, we introduce new qubits in $\cV_3$ and apply $S_3$ to measure the third qubit. Similarly, following Eq.~\eqref{eq:sim}, we can alternate between applying $\cV_{i}$ and processing $S_{i}$, calculate the conditional probability, and sample from each qubit accordingly. The procedure described above is formally expressed in Algorithm~\ref{alg:1d} and depicted in Fig.~\ref{fig:V}.

		We examine the space and time complexity as follows: The algorithm requires simulating $(t + 1)$ qubits simultaneously, which takes up $O(2^{2t})$ space complexity on a classical computer. On $(t + 1)$ qubits, simulating each two-qubit gate, each single-qubit noise channel, or each measurement takes up to $O(2^{2t})$ time. Given that the total number of gates with noise is within $O(n t)$, the time complexity to calculate all conditional probabilities, given previous sampling outcomes, is $O(2^{2t} nt)$. Our algorithm operates in $O(2^{2t} nt)$ computational time, with space complexity $O(2^{2t})$, thereby proving Lemma \ref{lem:simulate_1d}.
	\end{proof}

	Note that our analysis and the resulting Lemma \ref{lem:simulate_1d} are not restricted to any specific type of noise channel. Our results apply to noisy quantum devices with $O\left(\log (n)\right)$ depth under various forms of local noise, including the noiseless case, in contrast to previous works \cite{aharonov_limitations_1996, chen_complexity_2023, de_palma_limitations_2023}. When focusing on the single-qubit depolarizing noise channel, we can then prove Theorem~\ref{thm:1dcomplimit} by combining Theorem~\ref{thm:2_formal} with Lemma \ref{lem:simulate_1d} and substituting $t$ in the latter with $\frac{1}{|\log(\mu)|} (\log(T(n)) + 5)$.

	We consider 1D and 2D gate locality, depicted in Fig.~\ref{fig:topo}. The results are summarized and compared with the numerical results from the special case of random noisy circuits \cite{noh_efficient_2020,li_entanglement_2023} in Table~\ref{tab:boundscorr}.

	\begin{figure}[hbtp!]
		\centering
		\subfigure[One-dimensional lattice]{
			\begin{minipage}[t]{0.33\linewidth}
				\centering
				\raisebox{0.8cm}{\begin{tikzpicture}[scale=0.8]
						\coordinate (O) at (0,0);
						\foreach \x in {0,1,2,3,4} {
							\filldraw[black] (\x,1) circle (0.1);
						}
						\foreach \x [count=\xi] in {1,2,3,4} {
							\draw (\xi-1,1) -- (\xi,1);
						}
				\end{tikzpicture}}
			\end{minipage}%
		}%
		\subfigure[Two-dimensional lattice]{
			\begin{minipage}[t]{0.33\linewidth}
				\centering
				\begin{tikzpicture}[scale=0.8]
					\foreach \x in {0,1,2} {
						\foreach \y in {0,1,2} {
							\filldraw[black] (\x,\y) circle (0.1);
						}
					}
					\foreach \x in {0,1,2} {
						\foreach \y in {0,1} {
							\pgfmathtruncatemacro{\nexty}{\y+1}
							\draw (\x,\y) -- (\x,\nexty);
						}
					}
					\foreach \x in {0,1} {
						\foreach \y in {0,1,2} {
							\pgfmathtruncatemacro{\nextx}{\x+1}
							\draw (\x,\y) -- (\nextx,\y);
						}
					}
				\end{tikzpicture}
			\end{minipage}%
		}%
		\caption{Two typical topology setups for noisy quantum devices are studied in this work. The filled circles denote qubits. The lines are connections of those qubits where quantum gates can be placed.}
		\label{fig:topo}
	\end{figure}

	\begin{table}[hbtp!]
		\centering
		\caption{Bounds of quantum mutual information and relative entropy of entanglement in a noisy quantum device, when the qubit number is large. For 2D lattices, $A$ is a square with the size of $(n^\frac{1}{2}, n^\frac{1}{2})$. The small $n$ case is ignored in this table but considered in the theorems. Previous works numerically study the entanglement production of noisy circuits with random two-qubit gates arranged in a brick-wise architecture, which is a special case of our model. We list their results for comparison.}
		\label{tab:boundscorr}
		\begin{tabular}{c|c|c|c}
			Topology & Bounds & Scaling & Random noisy circuits \cite{noh_efficient_2020,li_entanglement_2023} \\
			\hline
			1D & $2\frac{\log (n)}{|\log(\mu)|}$  & $O(\log (n))$ & O(1) \\
			2D & $\frac{2\log(\frac{n}{2})}{|\log(\mu)|} n^{\frac{1}{2}}$ & $O(n^{\frac{1}{2}} \log(n))$ & $O(n^{\frac{1}{2}})$
		\end{tabular}
	\end{table}

	A key observation is that gates act on the qubits in a spatial locality, respecting the 1D chain topology. Consider a qubit in the chain within $t$ layers. Its interaction should be restrained within a distance of $t$. The restraint can be understood through a physical picture of a light cone, illustrated in Fig.~\ref{fig:lightcone}. The locality will limit entanglement spreading and the bipartite entanglement between halves of the chain. We develop the intuition into the following lemma.
	\begin{figure}[hbtp!]
		\centering
		\includegraphics[width=0.3\textwidth]{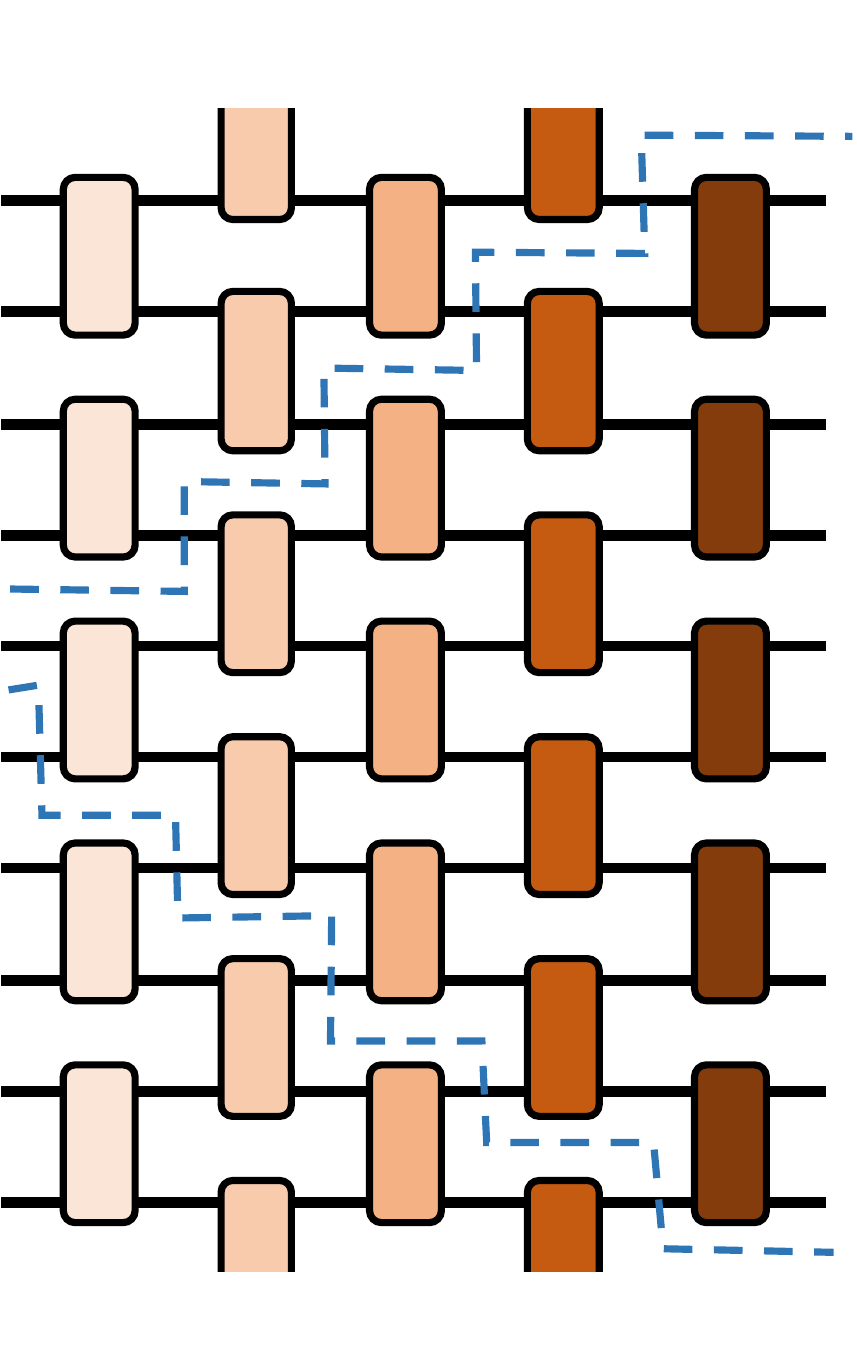}
		\caption{The physical picture of the light cone in one dimension. Dot lines are boundaries of the light cone. We use a specific brickwise architecture for illustration, which is not required. The color of the layers gets darker with increasing depth, representing that the state is approaching $\sigma_0$ according to Lemma \ref{lem:entropy}.}
		\label{fig:lightcone}
	\end{figure}

	\revise{
		\begin{lemma} \label{lem:1dlight}
			For a 1D lattice, namely a chain, at a fixed number of layers $t$, bipartite entanglement $E(\rho(t))$ between a contiguous subsystem $A$ and its complement part $\bar{A}$ is upper bounded by
			\begin{itemize}
				\item $4t$ if $A$ is in the bulk of the chain;
				\item $2t$ if $A$ is not in the bulk, i.e., contains one end of the chain.
			\end{itemize}
			The upper bounds hold for any quantities $E$ that do not increase under local operations and are upper bounded by the system size.
		\end{lemma}
		\begin{proof}
			Recall the output state is given by,
			\begin{equation} \label{eq:state_again}
				\rho(t) = \Lambda \circ \mathcal{U}_{t} \circ \Lambda \circ \cdots \circ \Lambda \circ \mathcal{U}_{2} \circ \Lambda \circ \mathcal{U}_{1} (\ketbra{0}^{\otimes n}),
			\end{equation}

			We will prove the lemma by iteratively reducing layers of gates and noise channels. The reduction is based on the fact that local operations cannot increase entanglement. The states in iterations are denoted by $\rho^\prime (k)$ where $k=t, t-1,\cdots, 1$ is the decreasing layer number. We start from the last layer, i.e., the $t$-th layer.

			Now, we divide gates in $\mathcal{U}_{t}$ into two kinds: $\mathcal{U}^{(1)}_{t}$ contains gates that act across $A$ and $\bar{A}$; while $\mathcal{U}^{(2)}_{t}$ contains other gates, which act on two qubits both in a same subsystem. Importantly, $\mathcal{U}^{(2)}_{t}$ is a local operation concerning the partition and has no contribution to entanglement. We construct the following state by reducing $\Lambda$ and $\mathcal{U}^{(2)}_{t}$,
			\begin{equation}
				\rho^\prime (t) =  \mathcal{U}^{(1)}_{t} (\rho(t-1)), \label{eq:reduced}
			\end{equation}
			satisfying $\rho(t) =  \Lambda \circ \mathcal{U}^{(2)}_{t} (\rho^{\prime}(t))$. It can be converted to $\rho(t)$ by local operations, thus $E(\rho (t)) \le E(\rho^\prime(t))$. And by the definition of $\mathcal{U}^{(1)}_{t}$, the size of its support is bounded by $|\sup(\mathcal{U}^{(1)}_{t})| \le 4$ if $A$ is in the bulk of the chain, and $|\sup(\mathcal{U}^{(1)}_{t})| \le 2$ if $A$ contains one end of the chain. The difference is due to the number of connections across $A$ and $\bar{A}$.

			Similarly, to reduce the ($t-1$)th layer, we again divided $\mathcal{U}_{t-1}$ into two kinds. This time, the gates and noise to be reduced must satisfy an additional condition that they should commute with $\mathcal{U}^{(1)}_{t}$, which requires that they do not overlap with $\sup \mathcal{U}^{(1)}_{t}$ in general.
			\begin{equation}
				\rho^\prime (t-1) =  \mathcal{U}^{(1)}_{t} \circ \Lambda_{\sup(\mathcal{U}^{(1)}_{t}) \cup \sup(\mathcal{U}^{(1)}_{t-1})} \circ \mathcal{U}^{(1)}_{t-1} (\rho(t-2)) = \mathcal{U}^{\prime}_{t-1} (\rho(t-2)), \label{eq:reduced2}
			\end{equation}
			where the combined channel $\mathcal{U}^{\prime}_{t-1}$ satisfying $|\sup(\mathcal{U}^{\prime}_{t-1})| \le 8$ if $A$ is in the bulk of the chain and $|\sup(\mathcal{U}^{\prime}_{t-1})| \le 4$ if $A$ contains one end of the chain. Local operations can again convert $\rho^\prime (t-1)$ to $\rho^\prime (t)$ so that $E(\rho^\prime (t)) \le E(\rho^\prime (t-1))$.

			We follow the procedure and iteratively reduce all layers of gates. Formally, we define the following notions at the $k$ layers, i.e., the $n+1-k$ step of the iteration:
			\begin{itemize}
				\item $\mathcal{U}^{(1)}_{k}$: gates whose support overlaps with $\sup(\mathcal{U}^{(1)}_{m})$ for some $m > k$ or act across $A$ and $\bar{A}$;
				\item $\mathcal{U}^{(2)}_{k}$: other gates in $\mathcal{U}_{k}$, which can be reduced;
				\item $\mathcal{U}^{\prime}_{k} = \mathcal{U}^{(1)}_{k} \circ \Lambda_{\sup(\mathcal{U}^{(1)}_{k}) \cup \sup(\mathcal{U}^{\prime}_{k+1})} \circ \mathcal{U}^{\prime}_{k+1}$: the combined channel after reducing the $k$th layer;
				\item $\rho^\prime (k) =  \mathcal{U}^{\prime}_{k} (\rho(k-1))$: the resulting state after reducing the $\mathcal{U}^{(2)}_{k}$.
			\end{itemize}
			By this construction, we have $\rho^\prime (k) = \Lambda_{\overline{\sup(\mathcal{U}^{\prime}_{k})}} \circ \mathcal{U}^{(2)}_{k}(\rho^\prime (k-1))$ and $E(\rho^\prime (k)) \le E(\rho^\prime (k-1))$ because the reduced part $\Lambda_{\overline{\sup(\mathcal{U}^{\prime}_{k})}}$ and $\mathcal{U}^{(2)}_{k}$ and $\Lambda$ cannot increase any entanglement.
			Also, the size of the support of $\mathcal{U}^{\prime}_{k}$ is bounded by $|\sup(\mathcal{U}^{\prime}_{k-1})| \le |\sup(\mathcal{U}^{\prime}_{k})| + 4$ if $A$ is in the bulk of the chain and $|\sup(\mathcal{U}^{\prime}_{k})| \le |\sup(\mathcal{U}^{\prime}_{k-1})| + 2$ if $A$ contains one end of the chain.

			After $t$ iterations, all layers are reduced, and we reach the first layer. The final resulting state is $\rho^\prime (1) =  \mathcal{U}^{\prime}_{1} (\ketbra{0}^{\otimes n})$ with $|\sup(\mathcal{U}^{\prime}_{1})| \le 4t$ if $A$ is in the bulk of the chain and $|\sup(\mathcal{U}^{\prime}_{1})| \le 2t$ if $A$ contains one end of the chain. With the assumption on the entanglement measure upper bounded by the system size, we conclude that $E(\rho(t)) \le E(\rho^\prime (1))$ and obtain the final result. The whole procedure is illustrated in Fig.~\ref{fig:reduce}.

			\begin{figure}[hbtp!]
				\centering
				\subfigure[The original circuit without reduction, producing $\rho(t)$ where $t=3$.]{
					\begin{minipage}[t]{0.4\linewidth}
						\centering
						\begin{quantikz}[row sep=0.8cm,column sep=0.1cm]
							& \gate[wires=2]{} & \gate{\Lambda_1} & \qw              & \gate{\Lambda_1} & \gate[wires=2]{} & \gate{\Lambda_1}  & \qw \\
							&                  & \gate{\Lambda_1} & \gate[wires=2]{} & \gate{\Lambda_1} &                  & \gate{\Lambda_1}  & \qw \\
							& \gate[wires=2]{} & \gate{\Lambda_1} &                  & \gate{\Lambda_1} & \gate[wires=2]{} & \gate{\Lambda_1}  & \qw \\
							&                  & \gate{\Lambda_1} & \gate[wires=2]{} & \gate{\Lambda_1} &                  & \gate{\Lambda_1}  & \qw \\
							& \gate[wires=2]{} & \gate{\Lambda_1} &                  & \gate{\Lambda_1} & \gate[wires=2]{} & \gate{\Lambda_1}  & \qw \\
							&                  & \gate{\Lambda_1} & \qw              & \gate{\Lambda_1} &                  & \gate{\Lambda_1}  & \qw \\
						\end{quantikz}
					\end{minipage}%
				}
				\subfigure[After one iteration of reduction, producing $\rho^\prime (t=3)$ in Eq.~\eqref{eq:reduced}.]{
					\begin{minipage}[t]{0.4\linewidth}
						\centering
						\begin{quantikz}[row sep=0.8cm,column sep=0.1cm]
							& \gate[wires=2]{} & \gate{\Lambda_1} & \qw              & \gate{\Lambda_1} & \qw              & \qw \\
							&                  & \gate{\Lambda_1} & \gate[wires=2]{} & \gate{\Lambda_1} & \qw              & \qw \\
							& \gate[wires=2]{} & \gate{\Lambda_1} &                  & \gate{\Lambda_1} & \gate[wires=2]{\mathcal{U}^\prime_3} & \qw \\
							&                  & \gate{\Lambda_1} & \gate[wires=2]{} & \gate{\Lambda_1} &                  & \qw \\
							& \gate[wires=2]{} & \gate{\Lambda_1} &                  & \gate{\Lambda_1} & \qw              & \qw \\
							&                  & \gate{\Lambda_1} & \qw              & \gate{\Lambda_1} & \qw                 & \qw \\
						\end{quantikz}
					\end{minipage}%
				}%

				\subfigure[After two iterations of reduction, producing $\rho^\prime (t-1=2)$ in Eq.~\eqref{eq:reduced2}.]{
					\begin{minipage}[t]{0.4\linewidth}
						\centering
						\begin{quantikz}[row sep=0.8cm,column sep=0.1cm]
							& \gate[wires=2]{} & \gate{\Lambda_1} & \qw                                 & \qw \\
							&                  & \gate{\Lambda_1} & \gate[wires=4]{\mathcal{U}^\prime_2} & \qw \\
							& \gate[wires=2]{} & \gate{\Lambda_1} &                                     & \qw \\
							&                  & \gate{\Lambda_1} &                                     & \qw \\
							& \gate[wires=2]{} & \gate{\Lambda_1} &                                     & \qw \\
							&                  & \gate{\Lambda_1} & \qw                                 & \qw \\
						\end{quantikz}
					\end{minipage}
				}%
				\subfigure[Ater three iterations of reduction, producing $\rho^\prime (t-2=1)$.]{
					\begin{minipage}[t]{0.4\linewidth}
						\centering
						\begin{quantikz}[row sep=0.8cm,column sep=0.1cm]
							& \gate[wires=6]{\mathcal{U}^\prime_1} & \qw \\
							&                                     & \qw \\
							&                                     & \qw \\
							&                                     & \qw \\
							&                                     & \qw \\
							&                                     & \qw \\
						\end{quantikz}
					\end{minipage}
				}%
				\caption{The iterative reduction of layers of gates. We show a case of six qubits across one boundary between $A$ and $\bar{A}$ as a part of a larger quantum device for illustration. The half $A$ contains the upper three qubits, while $\bar{A}$ contains the rest. We show the iterative reduction by a series of subfigures.}
				\label{fig:reduce}
			\end{figure}
		\end{proof}

		Lemma~\ref{lem:1dlight} is a generalization of previously established bounds for the entanglement entropy of pure states \cite{bravyi_lieb-robinson_2006, van_acoleyen_entanglement_2013, vershynina_entanglement_2019, eisert_entangling_2021}. We derive this result for the mixed-state case, which is essentially the case for noisy quantum devices. Then, by combining the results with Lemma \ref{lem:entropy}, we derive the upper bounds of quantum relative entropy of entanglement and quantum mutual information as Theorem~\ref{the:1dent} in the main text.
		\begin{proof}[Theorem~\ref{the:1dent}'s proof]
			First, we show that both the quantum relative entropy of entanglement and quantum mutual information are upper bounded by  $D(\rho \| \sigma_0)$. For the quantum relative entropy of entanglement, this is because
			\begin{equation}
				E_{R}(A:\bar{A}) = \min_{\sigma \in {\rm SEP}} D(\rho \| \sigma) \le D(\rho \| \sigma_0),
			\end{equation}
			where ${\rm SEP}$ denotes the set of separable states over $A$ and $\bar{A}$ partition. For quantum mutual information, this is because
			\begin{equation}
				I(A : \bar{A}) = S(A) + S(\bar{A}) - S(\rho) \le n - S(\rho) = D(\rho \| \sigma_0).
			\end{equation}
			And also we know $E_{R}(A:\bar{A})$ is upper bounded by $I(A : \bar{A})$, because $E_{R}(A:\bar{A}) = \min_{\sigma \in {\rm SEP}} D(\rho \| \sigma) \le D(\rho \| \rho_A \otimes \rho_{\bar{A}}) = I(A : \bar{A})$. Combined with Lemma \ref{lem:entropy},
			\begin{equation}
				I(A:\bar{A}) \le D(\rho \| \sigma_0)\le n \mu^{t}.
			\end{equation}
			From Lemma \ref{lem:1dlight}, we have another upper bounds
			\begin{equation}
				I(A:\bar{A}) \le 2t.
			\end{equation}
			By combining the two upper bounds, we have
			\begin{equation}
				I(A:\bar{A})\le \min \{ n \mu^{t}, 2t \} \le 2 t^\star,
			\end{equation}
			where $t^\star$ is a real number that satisfied $n \mu^{t^\star} \ge 2t^\star$. Note that $t$ is an integer, but the variable $t^\star$ extends it to real-valued. If $t^\star \ge 1$, we will have $n \mu^{t^\star} \ge 1$ and $I(A:\bar{A}) \le 2t^\star \le \frac{2\log(\frac{n}{2})}{|\log(\mu)|}$; otherwise, $I(A:\bar{A}) \le 2t^\star < 2$.
		\end{proof}
	}

	Besides the entanglement between adjacent regions, we also investigate the entanglement between two distant parts in a noisy qubit chain. Consider two distant, contiguous regions, $A$ and $B$. We define their distance $d(A, B)$ as the shortest path connecting them in a given qubit connection topology. If no noise exists, the entanglement between the two arbitrarily far parts can reach the optimal value of $\min (|A|,|B|)$ after a depth more than their distance. However, when the device suffers from noise and its distance is far, such a depth cannot be reached before the system gets too noisy. Formally, we have the following theorem.

	\begin{theorem}
		In a 1D noisy quantum circuit, the entanglement between the two distant regions $A$ and $B$, with distance $d(A,B)$ is upper bounded by
		\begin{equation}
			E(A:B) \le (|A| + |B|) (1-p)^{2d(A,B)} + \frac{4 \mu}{1-\mu}.
		\end{equation}
	\end{theorem}
	\begin{proof}
		Following the previously used method, we consider the loss of information in the system $AB$. Unlike in the previous problem, $AB$ may gain information by interacting with the outside systems. $AB$ can only interact with the four qubits adjacent to their boundaries for each layer in a 1D chain. Therefore, four bits of information can be regained at most before depolarizing noise comes. The information loss in a layer will be
		\begin{equation} \label{eq:infolosssubapp}
			D(\rho(t+1)_{AB} \| (\sigma_0)_{AB}) \le \mu \left[D(\rho(t)_{AB} \| (\sigma_0)_{AB}) + 4\right].
		\end{equation}
		We rewrite it into
		\begin{equation}
			D(\rho(t+1)_{AB} \| (\sigma_0)_{AB}) - \frac{4 \mu}{1-\mu} \le \mu \left[ D(\rho(t)_{AB} \| (\sigma_0)_{AB}) - \frac{4 \mu}{1-\mu} \right].
		\end{equation}
		This suggest that $D(\rho(t)_{AB} \| (\sigma_0)_{AB}) - \frac{4 \mu}{1-\mu}$ undergoes an exponential decay. Note that $D(\rho(0)_{AB} \| (\sigma_0)_{AB}) = |A| + |B|$ takes the maximal value. We will have an exponentially decaying upper bound of entanglement with an extra residual term $\frac{4 \mu}{1-\mu}$.
		\begin{equation}
			E(A:B) \le D(\rho(t)_{AB} \| (\sigma_0)_{AB}) \le (|A| + |B|) \mu^{t} + \frac{4 \mu}{1-\mu}.
		\end{equation}

		From gate locality, we know that when $t < d(A, B)$, the entanglement will be strictly zero. To see this, we can still use the reduction techniques shown in Fig.~\ref{fig:reduce} and do the same tricks. Combining this with the bounds we just obtained, we will eventually have
		\begin{equation}
			E(A:B) \le (|A| + |B|) (1-p)^{2d(A,B)} + \frac{4 \mu}{1-\mu}.
		\end{equation}
	\end{proof}

	\revise{
		For 2D lattices, we take the lattice as a regular square, whose size is $(n^{1/2}, n^{1/2})$, and also derive the entanglement bounds below.
		\begin{theorem}
			In a 2D lattice, whose shape is $(n^{1/2}, n^{1/2})$, for both quantum mutual information and quantum relative entropy of entanglement, we have upper bounds
			\begin{equation} \label{eq:2dent}
				E(A:\bar{A}) \le \max \left\{ \frac{2\log(\frac{n}{2})}{|\log(\mu)|} n^{\frac{1}{2}}, 4 n^{\frac{1}{2}} \right\},
			\end{equation}
			where $E$ is entanglement measured by quantum mutual information or relative entropy of entanglement.
		\end{theorem}
		\begin{proof}
			Similar to the 1D case, the proof proceeds by bounding the size of the support of the final effective operator, $|\sup(\mathcal{U}^{\prime}_{1})|$, which bounds the entanglement $E(\rho(t))$. This support region is contained by the ``band'' of all qubits within a Manhattan distance, i.e., the sum of the horizontal and vertical distances, of $t$ from the boundary of the subsystem $A$. The area of this band is the difference between the area of a square with side length $n^{\frac{1}{2}} + t$ and the area of a square with side length $n^{\frac{1}{2}} - t$.

			So the size of the support of the reduced channel, at last, will be upper-bounded
			\begin{equation}
				|\sup(\mathcal{U}^{\prime}_{1})| \le (n^{\frac{1}{2}}+t)(n^{\frac{1}{2}}+t) - (n^{\frac{1}{2}}-t)(n^{\frac{1}{2}}-t) = 4tn^{\frac{1}{2}}.
			\end{equation}

			For both quantum mutual information and quantum relative entropy of entanglement, denoted by $E(A:\bar{A})$, we have the following two upper bounds,
			\begin{equation}
				\begin{split}
					E(A:\bar{A}) & \le n \mu^{t}, \\
					E(A:\bar{A}) & \le 4tn^{\frac{1}{2}}.
				\end{split}
			\end{equation}
			Combining the two bounds, we have
			\begin{equation}
				E(A:\bar{A}) \le \min \{ n \mu^{t}, 4t n^{\frac{1}{2}} \} \le 4t^\star n^{\frac{1}{2}},
			\end{equation}
			where $t^\star$ is a real number that satisfies $n \mu^t \ge 4n^\frac{1}{2}t$.
			If $t^\star \ge 1$, we will have $n \mu^{t^\star} \ge 4n^{\frac{1}{2}}$ and $t^\star \le \frac{\frac{1}{2}\log(\frac{n}{2})}{|\log(\mu)|}$, thus $I(A:\bar{A}) \le 4t^\star n^{\frac{1}{2}} \le  \frac{2\log(\frac{n}{2})}{|\log(\mu)|}$; otherwise, $I(A:\bar{A}) \le 4 \sqrt{n}$. The same bound also holds for the quantum relative entropy of entanglement $E_{R}(A:\bar{A}) \le I(A:\bar{A})$.
		\end{proof}
	}

	When $n$ is sufficiently large, the upper bound scales as its leading term as $n^\frac{1}{2} \log (n)$, exhibiting an area-law scaling with an additional logarithmic factor. As in the 1D case, the proposed entanglement scaling limits the simulation of 2D quantum systems.

	\textit{Note added.}---After posting our work on arXiv, we realized that a recent work had also obtained an upper bound on entanglement growth regarding noisy circuit depth in their Lemma~14 using a different approach \cite{baspin_lower_2023}.

	\section*{Data availability}
	The data supporting the findings of this study are available from the first author upon reasonable request.

	\section*{Code availability}
	The theoretical results of the manuscript are reproducible from the analytical formulas and derivations presented therein. Additional code is available from the first author upon reasonable request.

	\bibliography{./materials/bibNisq.bib}

	\section*{Acknowledgements}
	The authors acknowledge Dorit Aharonov, Michael Ben-Or, and Li Gao for their helpful discussion about the convergence of quantum states. The authors acknowledge Ryuji Takagi for his helpful discussion about the convergence bound in qubit systems \cite{kastoryano_quantum_2013} and about error mitigation \cite{takagi_fundamental_2022, takagi_universal_2023}. The authors also thank Libor Caha, Honghao Fu, Yizhi Huang, Chenxu Li, Weikang Li, Guoding Liu, Chendi Yang, Xiao Yuan, and Xingjian Zhang for their helpful discussion. This work was supported by the National Natural Science Foundation of China Grant No.~12174216, and the Innovation Program for Quantum Science and Technology Grant No.~2021ZD0300804.

	\section*{Author contributions}
	Y.Y. initialized the work. Y.Y. and Z.D. derived the analytical theory and wrote the main manuscript text under the supervision of X.M. J.C. derived the entanglement capacity bound presented in Results. All authors reviewed and revised the manuscript.

	\section*{Competing interests}
	The authors declare no competing interests.

\end{document}